\documentclass[11pt]{article}
\usepackage{subfigure}
\usepackage{epsfig,graphics}
\usepackage{amssymb,amsmath,amsthm,bm}
\usepackage{url}
\usepackage{float}
\newtheorem{theorem}{Theorem}[section]                   
\newtheorem{definition}{Definition}[section]                   

\newtheorem{lemma}{Lemma}[section]

\newtheorem{remark}{Remark}[section]

\date{}
\begin{document}

\title{Efficient iterative thresholding algorithms with functional feedbacks and convergence analysis}

\author{
Ningning Han\thanks{School of Mathematics, Tianjin University, Tianjin 300350, China.
\textsf{Email:} ning\underline{\hbox to 2mm{}}ninghan@tju.edu.cn.},
 Shidong Li \thanks{Department of Mathematics, San Francisco State University, San Francisco, CA94132, USA. \textsf{Email:} shidong@sfsu.edu.},
Zhanjie Song\thanks{Corresponding author. School of Mathematics, Tianjin University, Tianjin 300350, China.
\textsf{Email:} zhanjiesongtju@gmail.com}
}
\maketitle

\begin{abstract}
An accelerated class of adaptive scheme of iterative thresholding algorithms is studied analytically and empirically.  They are based on the feedback mechanism of the null space tuning techniques (NST+HT+FB).  The main contribution of this article is the accelerated convergence analysis and proofs with a variable/adaptive index selection and different feedback principles at each iteration. These convergence analysis require no longer a priori sparsity information $s$ of a signal.
It is shown that uniform recovery of all $s$-sparse signals from given linear measurements can be achieved under reasonable (preconditioned) restricted isometry conditions.  Accelerated convergence rate and improved convergence conditions are obtained by selecting an appropriate size of the index support per iteration.   The theoretical findings are sufficiently demonstrated and confirmed by extensive numerical experiments. It is also observed that the proposed algorithms have a clearly advantageous balance of efficiency, adaptivity and accuracy compared with all other state-of-the-art greedy iterative algorithms.
\end{abstract}


\section{Introduction}

The emerging sparse approximation or compressed sensing \cite{Donoho,Can,Cands} has broken through the traditional notion of the Nyquist-Shannon sampling theorem \cite{Nyquist} by exploiting the compressibility or sparsity of sparse signals and nonlinear optimization techniques.  The compressed sensing problem is generally to solve a highly under-sampled or under-determined inverse problem
\begin{equation}\label{eq1}
y=Ax,
\end{equation}
given possibly linear measurements $y$, where $A\in\mathbb{C}^{M\times N}$ is the sampling matrix with $M\ll N$, and $x$ is the $N$-dimensional unknown signal with only $s\ll N $ nonzero coefficients. Solving linear inverse problems to find sparse solutions arises in a wide range of applications in signal and image processing \cite{JZhang}-\cite{ZYang}. Other applications are also found in low-rank tensor recovery  \cite{ECandes,Recht}, \cite{Liu,Gandy} through sparse approximation, and so on.

A naive way of tackling (\ref{eq1}) is to solve the combinatorial optimization problem for the sparsest solution by minimizing the $\ell_{0}$ ``norm'', which is known to be NP-hard, and computationally intractable£¬
\begin{equation}\label{P0}
\min_x \|x\|_0, \ \ \text{subj. to}\ \ y=Ax.
\end{equation}

Various computationally efficient algorithms for solving (\ref{eq1}) have been extensively studied. A large majority of algorithms are based on two strategies: convex or non-convex relaxations and greedy iterative algorithms.
The renowned advance of relaxations is to replace the optimization problem with the $\ell_{0}$ norm by the $\ell_{p}$ norms $(0<p\leq 1)$, namely,
\begin{equation}\label{Pp}
\min_x \|x\|_p, \ \  \text{subj. to}\ \ y=Ax, \ \ 0<p\leq 1.
\end{equation}
A symbolic work, known as basis pursuit (BP) \cite{SSChen}, is a typical convex relaxation by finding a solution of (\ref{Pp}) through $\ell_{1}$ minimization ($p=1$). Readers are referred to a series of articles focused on the theoretical analysis of the $\ell_{1}$ minimization approach. e.g., \cite{EJCand}-\cite{ZhangR}.

Note that $\ell_{1}$ norm may not accurately approximate $\ell_{0}$ norm, many scholars (see \cite{Chartrand}-\cite{LZhang}) have tried to provide solutions through further relaxations with non-convex $\ell_{p}$ norm. $0<p<1$, which, in some cases, do approximate $\ell_{0}$ norm better. In addition, weighted $\ell_{1}$ minimizations (see \cite{Ming}-\cite{YBZhao}) are another class of techniques designed to penalize more dramatically near-zero coefficients.  It is seen that weighted $\ell_{1}$ techniques do enhance the sparsity selection capacity and improve the signal recovery performance.

Similarly, an iterative re-weighted least square (IRLS) algorithm \cite{Daubechies} using a weighted $\ell_{2}$ minimization approach is another (nearly hidden) highlight among iterative algorithms.  It is noteworthy to emphasize that the IRLS is much less computationally demanding than that of the weighted $\ell_{1}$ techniques, which can be much more effective in large system applications.

There has also been a development of tail $\ell_{1}$ algorithms and analytical analysis \cite{LaiCK} as well. Theoretical analysis and experiment testing have confirmed the effectiveness of the tail $\ell_{1}$ algorithm, especially for sparse signals with large and near spark-level sparsity.  One of the notable results in \cite{LaiCK} is a measure theoretical uniqueness for the sparsest solution of (\ref{eq1}) when the sparsity $s$ satisfying $m/2 < s < m$, where $m$ is the spark of $A$. Here, the {\em spark} of $A$ stands for the least number of columns of $A$ that are linearly dependent.

Note that the uniquely sparsest solution of (\ref{P0}) in linear algebra only exists for $s < m/2$. The measure theoretical uniqueness theorem states that, when $m/2 < s < m$, the $\ell_0$ problem (\ref{P0}) still has unique solution in all $s$-sparse coordination spaces up to a measure of 0.  And, it is seen in \cite{LaiCK} that the tail $\ell_{1}$ algorithm has the capacity to recover sparse signals at the spark-level-sparsity $m/2 < s < m$ with probability/measure 1, whereas traditional $\ell_1$ based BP techniques is shown to must-fail for $m/2 < s < m$  \cite{LaiCK}, at least for real matrices $A$.

Greedy algorithms are another class of popular approaches, which find sparse solutions with, sometimes, considerably low computational complexity. A representative algorithm is the orthogonal matching pursuit (OMP) \cite{TroppJA,hongk}. At each iteration, the main principle of OMP is to exploit a column of sensing matrix $A$ maximally correlated with the residual, and add its corresponding index into the current support set. A new solution is then obtained by an orthogonal projection of the measurements onto the indexed column subspace of $A$ and finally to compute a new residual for the next iteration \cite{TroppJA,hongk}. Algorithms in the same category including regularized OMP (ROMP) \cite{romp}, stagewise OMP (StOMP) \cite{stomp}, subspace pursuit (SP) \cite{sp}, and compressive sampling matching pursuit (CoSaMP) \cite{cosamp}. An extension of OMP, named GOMP \cite{gOMP}, selects multiple correct indices per iteration so that the algorithm can terminate with much smaller number of iterations.

The simplest greedy approach is the iterative hard thresholding (IHT) \cite{IHTF,IHTA}, which keeps the largest several entries (in magnitude) of a vector and sets others to zeros. As important variations, normalized IHT (NIHT) \cite{NIHT} is proposed by considering an optimal step size. There is also an accelerated IHT (AIHT) \cite{AIHT}.  With the iteration $\widetilde{x}^{k+1}=H_{s}(x^{k}+\mu A^{\ast}(y-Ax^{k}))$, where $H_{s}$ is the hard thresholding operator that sets all but the largest (in magnitude) $s$ elements of a vector to zero, instead of continuing the iterative process with $\widetilde{x}^{k}$, AIHT finds $x^{k+1}$ that satisfies two conditions: $x^{k+1}$ is $s$-sparse and $x^{k+1}$ satisfies $\|y-Ax^{k+1}\|_{2}\leq\|y-A\widetilde{x}^{k+1}\|_{2}$. These variations lead to improved recovery capability and faster convergence speed than that of the traditional IHT. In addition, a conjugate gradient iterative hard thresholding (CGIHT) algorithm \cite{JBlanchard} is also seen to balance the low complexity per iteration of hard thresholding algorithms with the fast asymptotic convergence rate by using the conjugate gradient strategy. The theoretical analysis of CGIHT is provided in \cite{CGIHT}.

Despite the favorable results, algorithms based on IHT possess a drawback in that a prior knowledge of the sparsity level $s$ is generally required, and that the algorithms function well only for relatively smaller $s$.  Very recently, Foucart et al. propose a hard thresholding pursuit algorithm (HTP) \cite{Foucartone} and a graded hard thresholding pursuit algorithm (GHTP) \cite{Bouchot}, which can be regarded as a hybrid of IHT and CoSaMP. In \cite{GHTP}, some generalizations of HTP improve the speed performance by optimizing the size of support per iteration.

By an energy ``relocation'' strategy,  an iterative null space tuning algorithm with hard thresholding and feedbacks (NST+HT+FB) \cite{li2014fast} is proposed to find sparse solutions, aiming at faster convergence rate and greater recovery capacity. In fact, it is shown that NST+HT+FB converges in finite many steps \cite{li2014fast}.  The feedback mechanism in NST+HT+FB is to feed/relocate the ``tail'' contribution $A_{T_{k}^{c}}x_{T_{k}^{c}}^{k}$ of the measurement $y$ back to the thresholding support $T_k$:
\begin{equation} \label{eq8}
\text{(NST+HT+FB)}\ \ \ \
\left\{ \begin{aligned}
         \begin{aligned}
         &u_{T_{k}}^{k}=x_{T_{k}}^{k}+(A_{T_{k}}^{\ast}A_{T_{k}})^{-1}A_{T_{k}}^{\ast}A_{T_{k}^{c}}x_{T_{k}^{c}}^{k}, \\
         &u_{T_{k}^{c}}^{k}=0, \\
         &x^{k+1}=x^{k}+\mathbb{P}(u^{k}-x^{k}).\\
         \end{aligned}
         \end{aligned} \right.
\end{equation}
where $\mathbb{P}=I-A^{\ast}(AA^{\ast})^{-1}A$ is the orthogonal projection onto the $\ker(A)$, and can be computed off-line. If $x^{0}$ is feasible, then the iterative sequence $\{x^{k}\}$ is always feasible. Here the notations such as $A_T$ and $x_T$ stand for, respectively, columns of $A$ indexed by $T$, and components of $x$ indexed by $T$, following the convention in compressed sensing literatures and the matrix product rules.

Due to the feasibility of the sequence $\{x^{k}\}$, the NST step $x^{k+1}=x^{k}+\mathbb{P}(u^{k}-x^{k})$ can be rewritten as $x^{k+1}=u^{k}+A^{\ast}(AA^{\ast})^{-1}(y-Au^{k})$.
\begin{center}
     \begin{tabular}{lp{100mm}}
      \hline
      \multicolumn{2}{l}{{\bf Algorithm}~$1$~~AdptNST+HT+$f$-FB}\\ 
      \hline
      &{\bf Input:} $A$, $y$, $\epsilon$, $f(\cdot)$, $K$;\\
      &{\bf Output:} $u$;\\
      &{\bf Initialize:} $k=0$, $u^{0}=0$;\\
      &{\bf While}$\|y-Au^{k}\|_{2}>\epsilon$ \text{and} $k<K$ {\bf do}\\
      &~~~~~~~~~~$x^{k+1}=x^{k}+\mathbb{P}(u^{k}-x^{k})$; \\
      &~~~~~~~~~~$u_{T_{k}}^{k}=x_{T_{k}}^{k}+(A_{T_{k}}^{\ast}A_{T_{k}})^{-1}A_{T_{k}}^{\ast}A_{T^{c}_{k}}x_{T^{c}_{k}}^{k}$, $|T_{k}|=f(k)$.\\
      &~~~~~~~~~~$u_{T^{c}_{k}}^{k}=0$;\\
      &~~~~~~~~~~$k=k+1$;\\
      &{\bf end while};\\
     \hline
     \end{tabular}
\end{center}
Note that $|T_{k}|=s$, NST+HT+FB produces a sequence $\{u^{k}\}$ of $s$-sparse signals. If preconditioned restricted isometry constant \cite{li2014fast} and the restricted isometry constant of $A$ satisfies \cite{li2014fast}
 \begin{equation}\label{convercondition}
\delta_{2s}+\sqrt{2}\gamma_{3s}<1,
\end{equation}
then the sequence of $\{u^{k}\}$ generated by NST+HT+FB converges to the real solution $x$ \cite{li2014fast} rather rapidly.  As mentioned, it is shown in \cite{li2014fast} that NST+HT+FB converges in finite many steps.
Like other iterative thresholding algorithms, NST+HT+FB at its original form also assumes a prior knowledge of the sparsity $s$ and the convergence condition (\ref{convercondition}) is still stronger than needed.

Motivated by topics about selection of indices per iteration in greedy algorithms, we introduce a generalization of NST+HT+FB, referred to as ApdtNST+HT+$f$-FB, in terms of adaptive but fixed functional index set selections per iteration. Specifically, the cardinality of indices selected per iteration $|T(k)| = f(k)$ at the $k^{th}$ iteration. The AdptNST+HT+$f$-FB algorithm can be established as follows:
 \begin{equation} \label{eq10}
(\text{AdptNST+HT+$f$-FB})\ \ \ \
\left\{ \begin{aligned}
         \begin{aligned}
         &\mu_{T_{k}}^{k}=x_{T_{k}}^{k}+(A_{T_{k}}^{\ast}A_{T_{k}})^{-1}A_{T_{k}}^{\ast}A_{T_{k}^{c}}x_{T_{k}^{c}}^{k}, \\
         &\mu_{T_{k}^{c}}^{k}=0, \\
         &x^{k+1}=x^{k}+\mathbb{P}(u^{k}-x^{k}),\\
         \end{aligned}
         \end{aligned} \right.
\end{equation}
where $|T_{k}|=f(k)$. Note that the constant function $f(k)=s$ corresponds to NST+HT+FB (\ref{eq8}).
Since $|T_{k}|=f(k)$, AdptNST+HT+$f$-FB constructs a sequence $\{u^{k}\}$ of $f(k)$-sparse signals per iteration.  Algorithm $1$ is the pseudo-code of AdptNST+HT+$f$-FB, where $K$ is the maximum number of iterations.

The main contribution of this article is to provide the detailed convergence analysis and proofs for the proposed class of algorithms combining hard thresholding, $f$-feedbacks and null space tuning. Note that the significant departure of this paper from the previous literature \cite{li2014fast} includes the followings. For one, the number of indices selected per iteration is $f(k)$ without requiring a prior knowledge or estimation of the sparsity level $s$.  This choice of $|T(k)|=f(k)$ is shown to improve the adaptivity and the speed of convergence.  For two, the general convergence theory is obtained for AdptNST+HT+$f$-FB.  Since AdptNST+HT+$f$-FB is reduced to NST+HT+FB by setting $f(k)=s$, the new convergence condition for AdptNST+HT+$f$-FB also improves that of NST+HT+FB in \cite{li2014fast}, as will be discussed in {\em Section \ref{secConvergence}}.

The AdptNST+HT+$f$-FB algorithm is compared with other relevant algorithms previously mentioned empirically through extension numerical tests.  These numerical experiments show that the AdptNST+HT+$f$-FB algorithm is among the most advanced and effective recovery algorithms.  It is also seen that AdptNST+HT+$f$-FB algorithm has a clearly advantageous balance of efficiency, adaptivity and accuracy compared with all other state-of-the-art greedy iterative algorithms.

For clarity, notations are used as follows in this article. $S$ is the true support of $s$-sparse vector $x$. $x_{T}$ is the restriction of a vector $x$ to an index set $T$. We denote by $T^{c}$ the complement set of $T$ in $\{1,2,\ldots,N\}$, and by $A_{T}$ the sub-matrix consisting of columns of $A$ indexed by $T$, respectively. $T\triangle T'$ is the symmetric difference of $T$ and $T'$, i.e., $T\triangle T' =(T\setminus T')\cup(T'\setminus T)$ and $|T|$ is the cardinality of set $T$.

\section{Main results}\label{secConvergence}
\subsection{Properties characterized via RIP (P-RIP)}
\begin{definition} \label{de1}
\cite{Can} For each integer $s = 1, 2,... ,$ the restricted isometry constant $\delta_{s}$ of a matrix $A$ is defined as the smallest number $\delta_{s}$ such that
\begin{equation*}\label{eq10}
(1-\delta_{s})\|x\|_{2}^{2}\leq\|Ax\|_{2}^{2}\leq(1+\delta_{s})\|x\|_{2}^{2},
\end{equation*}
holds for all $s$-sparse vectors $x$. Equivalently, it can be given by \cite{FoucatsBook}
\begin{equation*}\label{eq13}
\delta_{s}=\max \limits_{S\subset[N],|S|\leq s}\|A_{S}^{\ast}A_{S}-I\|_{2}.
\end{equation*}
\end{definition}

\begin{definition} \label{de2}
 \cite{li2014fast}. For each integer $s = 1, 2,... ,$ the preconditioned restricted isometry constant $\gamma_{s}$ of a matrix $A$ is defined as the smallest number $\gamma_{s}$ such that
\begin{equation*}\label{eq11}
(1-\gamma_{s})\|x\|_{2}^{2}\leq\|(AA^{\ast})^{-\frac{1}{2}}Ax\|_{2}^{2},
\end{equation*}
holds for all $s$-sparse vectors $x$.
In fact, the preconditioned restricted isometry constant $\gamma_{s}$ characterizes the restricted isometry property of the preconditioned matrix $(AA^{\ast})^{-\frac{1}{2}}A$. Since
\begin{equation*}
\|(AA^{\ast})^{-\frac{1}{2}}Ax\|_{2}\leq\|(AA^{\ast})^{-\frac{1}{2}}A\|_{2}\|x\|_{2}=\|x\|_{2},
\end{equation*}
$\gamma_{s}$ is actually the smallest number such that, for all $s$-sparse vectors $x$,
\begin{equation*}\label{eq12}
(1-\gamma_{s})\|x\|_{2}^{2}\leq\|(AA^{\ast})^{-\frac{1}{2}}Ax\|_{2}^{2}\leq(1+\gamma_{s})\|x\|_{2}^{2}.
\end{equation*}

Note that $\gamma_{s}(A)=\delta_{s}((AA^{\ast})^{-\frac{1}{2}}A)$. Evidently, for Parseval frames, since $AA^{\ast}=I$, $\gamma_{s}(A)=\delta_{s}(A)$.  Equivalently, $\gamma_{s}$ can also be given by
\begin{equation*}\label{eq14}
\gamma_{s}=\max \limits_{S\subset[N],|S|\leq s}\|A_{S}^{\ast}(AA^\ast)^{-1}A_{S}-I\|_{2}.
\end{equation*}
\end{definition}

\begin{lemma} \label{ert}
For $u,v\in\mathbb{C}^{N}$, if $|\text{supp} (u) \cup \text{supp} (v)|\leq t$, then $|\langle u,(I-A^{\ast}A)v\rangle|\leq\delta_{t}\|u\|_{2}\|v\|_{2}$. Suppose $|R \cup \text{supp} (v)|\leq t$, then $\|((I-A^{\ast}A)v)_{R}\|_{2}\leq\delta_{t}\|v\|_{2}$.
\end{lemma}

\begin{proof}
Indeed, setting $T=\text{supp} (v)\cup  \text{supp} (u)$, one has
\begin{equation}\label{ningyang}
\begin{aligned}
|\langle u,(I-A^{\ast}A)v\rangle| &=|\langle u_{T},v_{T}\rangle-\langle A_{T}u_{T},A_{T}v_{T}\rangle|\\
&=|\langle u_{T},(I-A_{T}^{\ast}A_{T})v_{T}\rangle|\\
&\leq\|u_{T}\|_{2}\|I-A_{T}^{\ast}A_{T}\|_{2}\|v_{T}\|_{2}\\
&\leq\delta_{t}\|u\|_{2}\|v\|_{2}.\\
\end{aligned}
\end{equation}
Using (\ref{ningyang}), we have
\begin{equation*}
\begin{aligned}
\|((I-A^{\ast}A)v)_{R}\|_{2}^{2} &=\langle((I-A^{\ast}A)v)_{R},(I-A^{\ast}A)v\rangle\\
&\leq\delta_{t}\|((I-A^{\ast}A)v)_{R}\|_{2}\|v\|_{2},
\end{aligned}
\end{equation*}
and it remains to simplify by solving $\|((I-A^{\ast}A)v)_{R}\|_{2}$ to obtain $\|((I-A^{\ast}A)v)_{R}\|_{2}\leq\delta_{t}\|v\|_{2}$.
\end{proof}

\begin{remark} \label{nin}
Let $\gamma_{t}$ be the P-RIP constant of $A$, i.e., $\gamma_{t}(A)=\delta_{t}((AA^{\ast})^{-\frac{1}{2}}A)$. For $u,v\in\mathbb{C}^{N}$, if $|\text{supp} (u) \cup \text{supp} (v)|\leq t$, then $|\langle u,(I-A^{\ast}(AA^{\ast})^{-1}A)v\rangle|\leq\gamma_{t}\|u\|_{2}\|v\|_{2}$.  Moreover, Suppose $|R \cup \text{supp} (v)|\leq t$, then $\|((I-A^{\ast}(AA^{\ast})^{-1}A)v)_{R}\|_{2}\leq\gamma_{t}\|v\|_{2}$.
\end{remark}

\begin{lemma}\label{dong}
For $\forall e\in\mathbb{C}^{M}$ and $|T|\leq t$, then $\|(A^{\ast}e)_{T}\|_{2}\leq\sqrt{1+\delta_{t}}\|e\|_{2}$.
\end{lemma}

\begin{proof}
\begin{equation*}
\begin{aligned}
\|(A^{\ast}e)_{T}\|_{2}^{2}&=\langle A^{\ast}e,(A^{\ast}e)_{T}\rangle\\
&=\langle e,A(A^{\ast}e)_{T}\rangle\\
&\leq\|e\|_{2}\|A(A^{\ast}e)_{T}\|_{2}\\
&\leq\|e\|_{2}\sqrt{1+\delta_{t}}\|(A^{\ast}e)_{T}\|_{2},\\
\end{aligned}
\end{equation*}
we have $\|(A^{\ast}e)_{T}\|_{2}\leq\sqrt{1+\delta_{t}}\|e\|_{2}$.
\end{proof}

\begin{remark}\label{guang}
 For $\forall$ $e\in\mathbb{C}^{M}$ and $|T|\leq t$, then $\|(A^{\ast}(AA^{\ast})^{-1}e)_{T}\|_{2}\leq\sqrt{1+\theta_{t}}\|e\|_{2}$, where $\theta_{t}(A)=\delta_{t}((AA^{\ast})^{-1}A)$.
\end{remark}

\begin{lemma}\label{h}
 Suppose that $y=Ax+e$, where $x\in\mathbb{C}^{N}$ is $s$-sparse with $S=$\text{supp}$(x)$ and $e\in\mathbb{C}^{M}$ is the measurement error. If $u'\in\mathbb{C}^{N}$ is $s'$-sparse and $T$ is an index set of $t\geq s$ largest absolute entries of $u'+A^{\ast}(AA^\ast)^{-1}(y-Au')$, then we have
\begin{equation*}
\|x_{T^{c}}\|_{2}\leq\sqrt{2}(\gamma_{s+s'+t}\|x-u'\|_{2}+\sqrt{1+\theta_{t+s}}\|e\|_{2}),
\end{equation*}
where $\theta_{s}(A)=\delta_{s}((AA^\ast)^{-1}A)$.
\end{lemma}

\begin{proof}
It is well known that

$\|[u'+A^{\ast}(AA^\ast)^{-1}(y-Au')]_{T}\|_{2}\geq\|[u'+A^{\ast}(AA^\ast)^{-1}(y-Au')]_{S}\|_{2}.$

Eliminating the common terms over $T\bigcap S$, one has

$\|[u'+A^{\ast}(AA^\ast)^{-1}(y-Au')]_{T\setminus S}\|_{2}\geq\|[u'+A^{\ast}(AA^\ast)^{-1}(y-Au')]_{S\setminus T}\|_{2}.$

For the left hand side,
\begin{equation*}
\begin{aligned}
\|[u'+A^{\ast}(AA^\ast)^{-1}(y-Au')]_{T\setminus S}\|_{2}&=\|[u'-x+A^{\ast}(AA^\ast)^{-1}(Ax+e-Au')]_{T\setminus S}\|_{2}\\
&=\|[(I-A^{\ast}(AA^\ast)^{-1}A)(u'-x)+A^{\ast}(AA^\ast)^{-1}e]_{T\setminus S}\|_{2}.\\
\end{aligned}
\end{equation*}

The right hand side satisfies
\begin{equation*}
\begin{aligned}
\|[u'+A^{\ast}(AA^\ast)^{-1}(y-Au')]_{S\setminus T}\|_{2}&=\|[u'+A^{\ast}(AA^\ast)^{-1}(Ax+e-Au')+x-x]_{S\setminus T}\|_{2}\\
&\geq\|x_{S\setminus T}\|_{2}-\|[(I-A^{\ast}(AA^\ast)^{-1}A)(u'-x)+A^{\ast}(AA^\ast)^{-1}e]_{S\setminus T}\|_{2}.\\
\end{aligned}
\end{equation*}
Consequently,
\begin{equation*}
\begin{aligned}
&\|x_{S\setminus T}\|_{2}\\
&\leq \|[(I-A^{\ast}(AA^\ast)^{-1}A)(u'-x)+A^{\ast}(AA^\ast)^{-1}e]_{S\setminus T}\|_{2}\\
&+\|[(I-A^{\ast}(AA^\ast)^{-1}A)(u'-x)+A^{\ast}(AA^\ast)^{-1}e]_{T\setminus S}\|_{2}\\
&\leq\sqrt{2}\|[(I-A^{\ast}(AA^\ast)^{-1}A)(u'-x)+A^{\ast}(AA^\ast)^{-1}e]_{T\triangle S}\|_{2}\\
&\leq\sqrt{2}\|[(I-A^{\ast}(AA^\ast)^{-1}A)(u'-x)]_{T\triangle S}\|_{2}+\sqrt{2}\|[A^{\ast}(AA^\ast)^{-1}e]_{T\triangle S}\|_{2}\\
&\leq\sqrt{2}(\gamma_{s+s'+t}\|x-u'\|_{2}+\sqrt{1+\theta_{t+s}}\|e\|_{2}).
\end{aligned}
\end{equation*}
Here in the last step {\em Remarks \ref{nin}} and {\em \ref{guang}} are applied.
\end{proof}

\begin{lemma}\label{lemm6}
 Suppose that $y=Ax+e$, where $x\in\mathbb{C}^{N}$ is $s$-sparse with $S=$\text{supp}$(x)$ and $e\in\mathbb{C}^{M}$ is the measurement error. Let $T=$\text{supp}$(x')$ and $|T|=t$. If $u'$ is the feedback of $x'$ that subjects to $u_{T}'=x_{T}'+(A_{T}^{\ast}A_{T})^{-1}A_{T}^{\ast}A_{T^{c}}x_{T^{c}}'$ and $u_{T^{c}}'=0$, then
\begin{equation*}
\|(x-u')\|_{2}\leq\frac{\sqrt{1+\delta_{t}}\|e\|_{2}}{1-\delta_{s+t}}+\frac{\|x_{T^{c}}\|_{2}}{\sqrt{1-\delta_{s+t}^{2}}}.
\end{equation*}
\end{lemma}

\begin{proof}
For any $z\in\mathbb{C}^{N}$ supported on $T$,
\begin{equation*}
\begin{aligned}
\langle Au'-y,Az\rangle&=\langle A_{T}x_{T}'+A_{T}( A_{T}^{\ast} A_{T})^{-1}A_{T}^{\ast} A_{T^{c}}x_{T^{c}}'-y,A_{T}z_{T}\rangle\\
&=\langle A_{T}^{\ast}(A_{T}x_{T}'+A_{T^{c}}x_{T^{c}}'-y),z_{T}\rangle\\
&=\langle A_{T}^{\ast}(Ax'-y),z_{T}\rangle\\
&=0.\\
\end{aligned}
\end{equation*}
The last step is due to the feasibility of $x'$, i.e., $y=Ax'$.
The inner product can also be written as

$\langle Au'-y,Az\rangle=\langle (Au'-Ax-e),Az\rangle=0$.

\noindent We then have,

$\langle (u'-x),A^{\ast}Az\rangle=\langle e,Az\rangle,~~\forall z\in\mathbb{C}^{N}$ supported on $T$.

\noindent Since $(u'-x)_{T}$ is supported on $T$, one has

$\langle (u'-x),A^{\ast}A(u'-x)_{T}\rangle=\langle e,A(u'-x)_{T}\rangle.$

\noindent Consequently,
\begin{equation*}
\begin{aligned}
\|(u'-x)_{T}\|_{2}^{2}&=\langle(u'-x),(u'-x)_{T}\rangle\\
&=|\langle (x-u'),(I-A^{\ast}A)(x-u')_{T}\rangle+\langle e,A(u'-x)_{T}\rangle|\\
&\leq\delta_{s+t}\|x-u'\|_{2}\|(x-u')_{T}\|_{2}+\sqrt{1+\delta_{t}}\|e\|_{2}\|(x-u')_{T}\|_{2}\\
\end{aligned}
\end{equation*}
Using  {\em Lemma \ref{ert}}, Cauchy-Schwarz inequality and {\em Definition \ref{de1}} can obtain the last inequality.
\noindent Now, we have

$\|(x-u')_{T}\|_{2}\leq\delta_{s+t}\|x-u'\|_{2}+\sqrt{1+\delta_{t}}\|e\|_{2}$.

\noindent It then follows that

\begin{equation*}
\begin{aligned}
\|(x-u')\|_{2}^{2}&=\|(x-u')_{T}\|_{2}^{2}+\|(x-u')_{T^{c}}\|_{2}^{2}\\
&\leq(\delta_{s+t}\|x-u'\|_{2}+\sqrt{1+\delta_{t}}\|e\|_{2})^{2}+\|x_{T^{c}}\|_{2}^{2}.\\
\end{aligned}
\end{equation*}

\noindent In other words,
\begin{equation*}
\begin{aligned}
(\sqrt{1-\delta_{s+t}^{2}}\|(x-u')\|_{2}-\frac{\delta_{s+t}\sqrt{1+\delta_{t}}}{\sqrt{1-\delta_{s+t}^{2}}}\|e\|_{2})^{2}
\leq\frac{1+\delta_{t}}{1-\delta_{s+t}^{2}}\|e\|_{2}^{2}+\|x_{T^{c}}\|_{2}^{2}.
\end{aligned}
\end{equation*}
\noindent It means that
\begin{equation*}
\begin{aligned}
\|(x-u')\|_{2}&\leq\frac{\delta_{s+t}\sqrt{1+\delta_{t}}\|e\|_{2}+\sqrt{(1+\delta_{t})\|e\|_{2}^{2}+({1-\delta_{s+t}^{2})\|x_{T^{c}}\|_{2}^{2}}}}{1-\delta_{s+t}^{2}}\\
&\leq\frac{\|x_{T^{c}}\|_{2}}{\sqrt{1-\delta_{s+t}^{2}}}+\frac{\sqrt{1+\delta_{t}}\|e\|_{2}}{1-\delta_{s+t}}.\\
\end{aligned}
\end{equation*}
\end{proof}

\begin{theorem}\label{ha}
Suppose $y=Ax+e$ with $s$-sparse signal $x$ and measurement error $e$. If the RIP and P-RIP constants of $A$ satisfy $2\gamma_{s+f(k-1)+f(k)}^{2}+\delta_{s+f(k)}^{2}<1$, then $\{u^{k}\}$ in AdptNST+HT+$f$-FB satisfies
\begin{equation*}\label{hhh}
\begin{aligned}
\|(x-u^{k})\|_{2}&\leq\sqrt{\frac{2\gamma_{s+f(k-1)+f(k)}^{2}}{(1-\delta_{s+f(k)}^{2})}}\|x-u^{k-1}\|_{2}+\left(\frac{\sqrt{1+\delta_{f(k)}}}{1-\delta_{s+f(k)}}+
\frac{\sqrt{2(1+\theta_{s+f(k)})}}{\sqrt{1-\delta_{s+f(k)}^{2}}}\right)|e\|_{2},~k\geq s.\\
\end{aligned}
\end{equation*}
\end{theorem}

\begin{proof}
 Applying {\em Lemma \ref{h}} to $u'=u^{k-1}$ and $T=T_{k}$ for $k\geq s$ gives rise to
\begin{equation*}
\|x_{T^{c}_{k}}\|\leq\sqrt{2}\left(\gamma_{s+f(k-1)+f(k)}\|x-u^{k-1}\|_{2}+\sqrt{1+\theta_{s+f(k)}}\|e\|_{2}\right),
\end{equation*}
\noindent and setting $u'=u^{k}$ and $T=T_{k}$ in {\em Lemma \ref{lemm6}}, one obtains
\begin{equation*}
\|(x-u^{k})\|_{2}\leq\frac{1}{\sqrt{(1-\delta_{s+f(k)}^{2})}}\|x_{T^{c}_{k}}\|_{2}+\frac{\sqrt{1+\delta_{f(k)}}}{1-\delta_{s+f(k)}}\|e\|_{2}.
\end{equation*}
\noindent Combining these two inequalities, we have
\begin{equation*}
\begin{aligned}
\|(x-u^{k})\|_{2}\leq\sqrt{\frac{2\gamma_{s+f(k-1)+f(k)}^{2}}{(1-\delta_{s+f(k)}^{2})}}\|x-u^{k-1}\|_{2}+\left(\frac{\sqrt{1+\delta_{f(k)}}}{1-\delta_{s+f(k)}}+
\frac{\sqrt{2(1+\theta_{s+f(k)})}}{\sqrt{1-\delta_{s+f(k)}^{2}}}\right)\|e\|_{2}.
\end{aligned}
\end{equation*}
\end{proof}

Through {\em Theorem \ref{ha}}, if the P-RIP and RIP constants of $A$ satisfy $2\gamma_{s+f(k-1)+f(k)}^{2}+\delta_{s+f(k)}^{2}<1$, then the sequence of $\{u_{k}\}$ generated by AdptNST+HT+$f$-FB converges to $x$.

\begin{remark}\label{hhhh1}
If a prior estimation of the sparsity $s$ is assumed known, then setting $f(k)=s$ yields the original NST+HT+FB, where the resulting sequence $u^{k}$ satisfies
\begin{equation*}
\begin{aligned}
\|x-u^{k}\|_{2}&\leq\sqrt{\frac{2\gamma_{3s}^{2}}{(1-\delta_{2s}^{2})}}\|x-u^{k-1}\|_{2}+\left(\frac{\sqrt{1+\delta_{s}}}{1-\delta_{2s}}+
\frac{\sqrt{2(1+\theta_{2s})}}{\sqrt{1-\delta_{2s}^{2}}}\right)\|e\|_{2},~k\geq1.
\end{aligned}
\end{equation*}
As shown in {\em Remark \ref{hhhh1}}, if the P-RIP and RIP constants of $A$ satisfy $\delta_{2s}^{2}+2\gamma_{3s}^{2}<1$, then the sequence of $\{u^{k}\}$ generated by NST+HT+FB converges to $x$. Compared to the condition $\delta_{2s}+\sqrt{2}\gamma_{3s}<1$ in \cite{li2014fast}, the condition in {\em Remark \ref{hhhh1}} is obvious improved. Furthermore, if $A$ is the Parseval frame, the P-RIP and RIP condition is relaxed to RIP condition, i.e., $\delta_{2s}^{2}+2\delta_{3s}^{2}<1$.
\end{remark}

\section{Numerical experiments}
In this section, we demonstrate through extensive numerical experiments the claim that the convergence speed of the proposed algorithm is improved by elaborating the number of indices $f(k)$ selected per iteration. We first compare the performances within the class of the iterative thresholding algorithms with feedbacks by taking, respectively, cardinality of support per iteration as $|T_{k}|=s$ (\text{NST+HT+FB}), $|T_{k}|=k$, $|T_{k}|=2k$, $|T_{k}|=4k$, $|T_{k}|=6k$ and $|T_{k}|=k^{2}$.   Then the overall performance of AdptNST+HT+$f$-FB in terms of execution-time and frequency of exact recovery is compared with state-of-the-art greedy iterative algorithms including accelerated iterative hard thresholding (AIHT) \cite{AIHT}, generalized orthogonal matching pursuit (GOMP) \cite{gOMP},  conjugate gradient iterative hard thresholding algorithm (CGIHT) \cite{CGIHT}, and graded hard thresholding pursuit (GHTP) \cite{Bouchot,GHTP}.

\begin{figure}[H]
\centering
\begin{tabular}{cc}
\includegraphics[width=6cm]{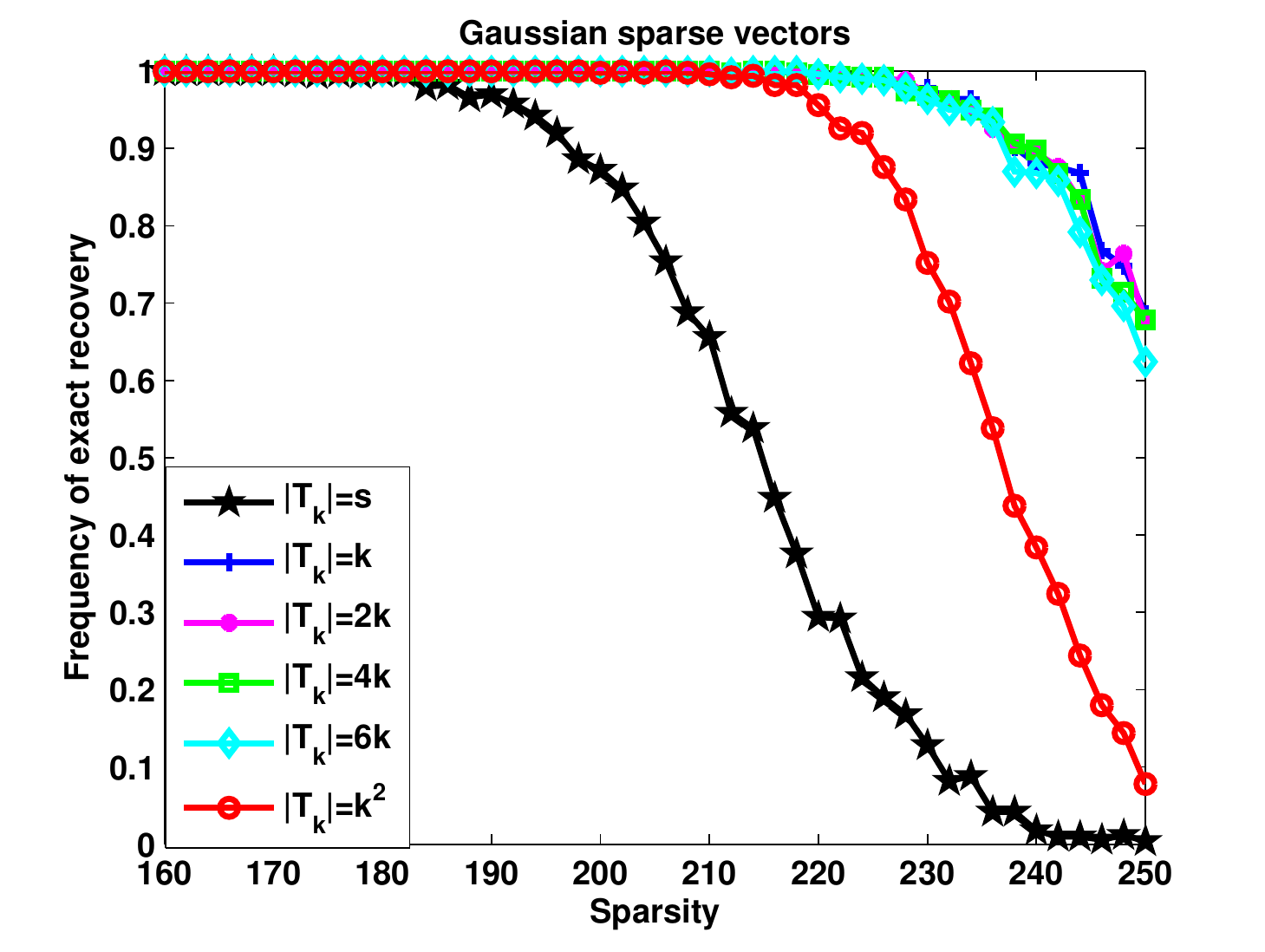}
&\includegraphics[width=6cm]{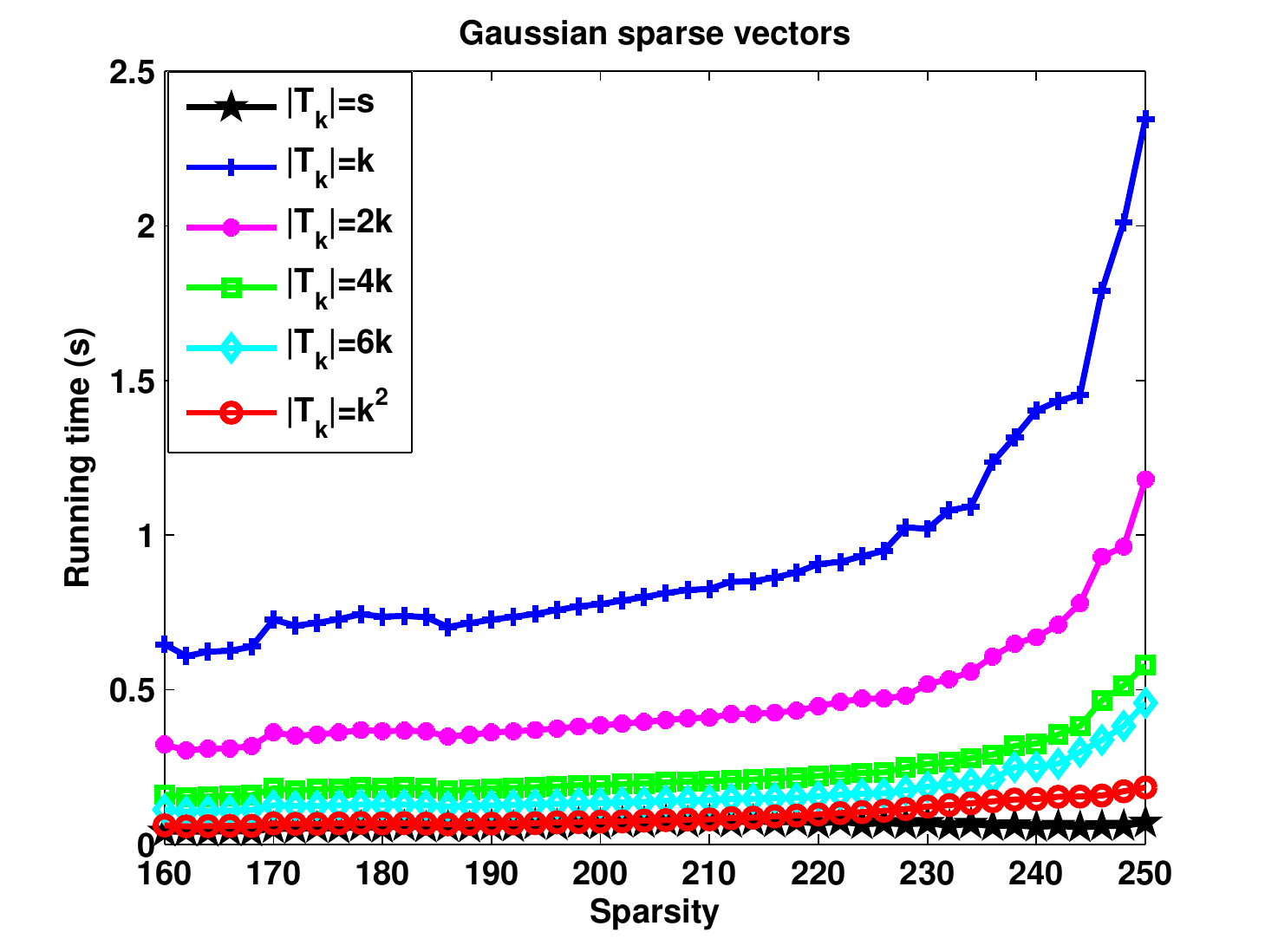}\\
\includegraphics[width=6cm]{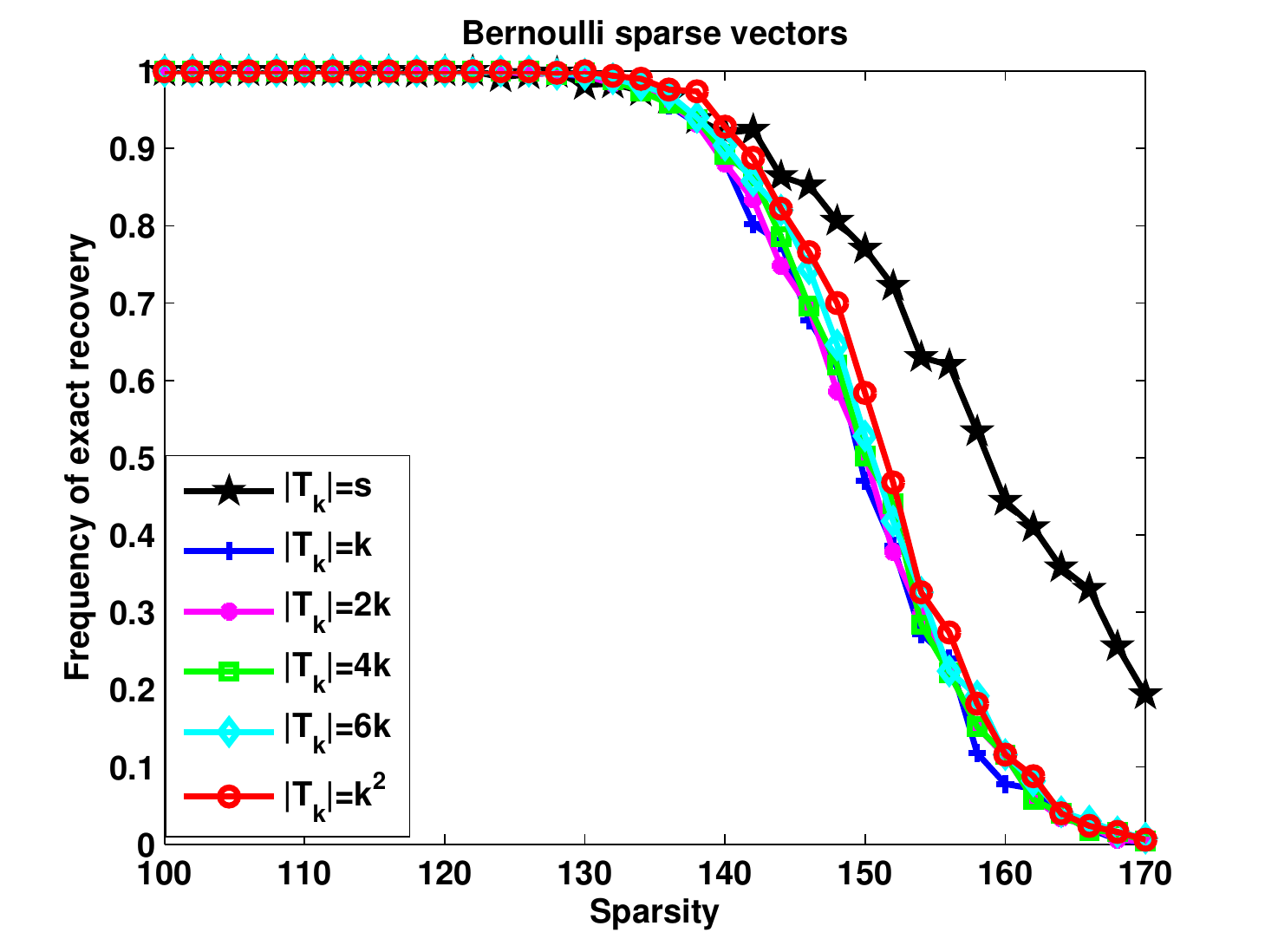}
&\includegraphics[width=6cm]{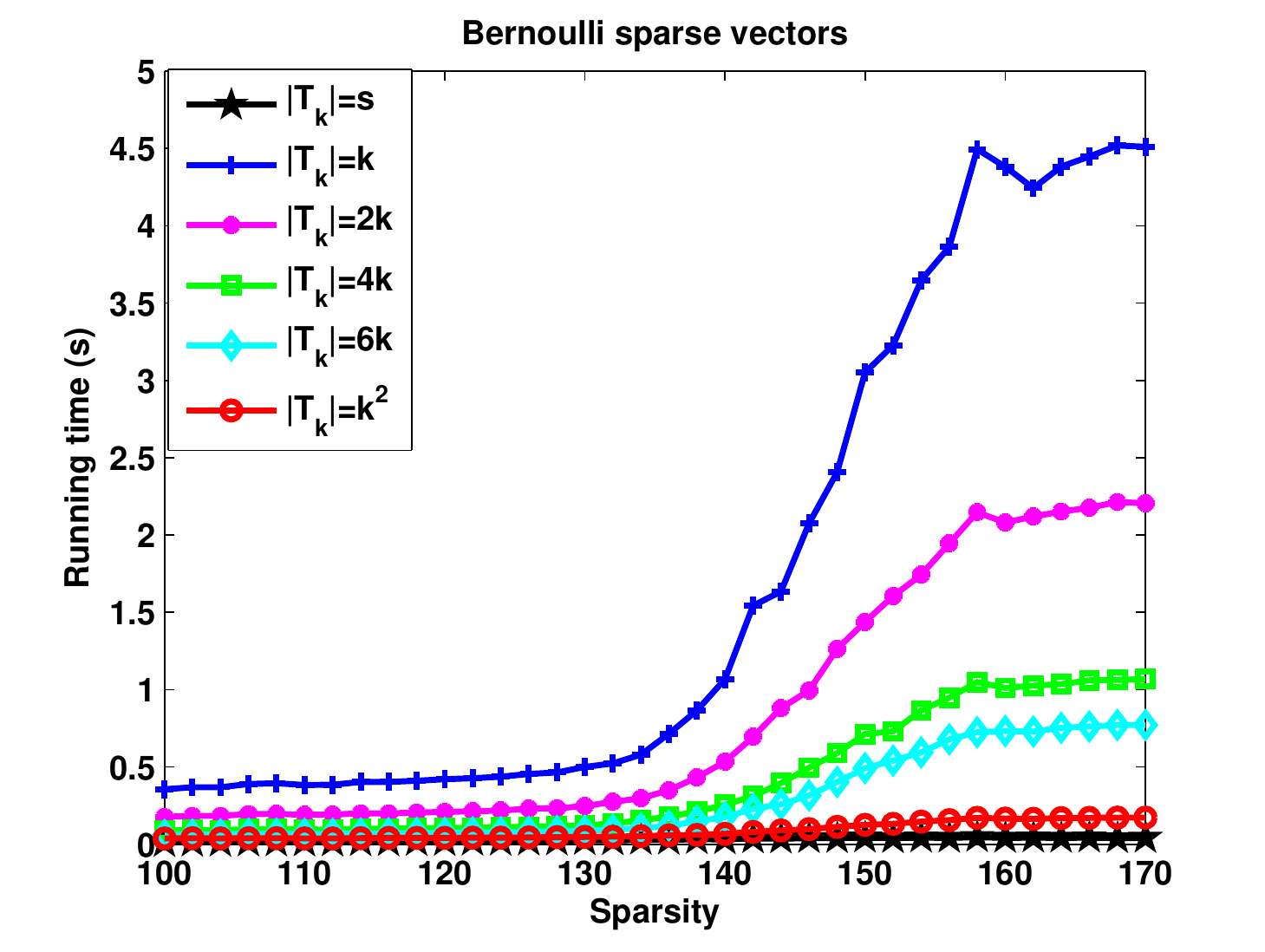}\\
\includegraphics[width=6cm]{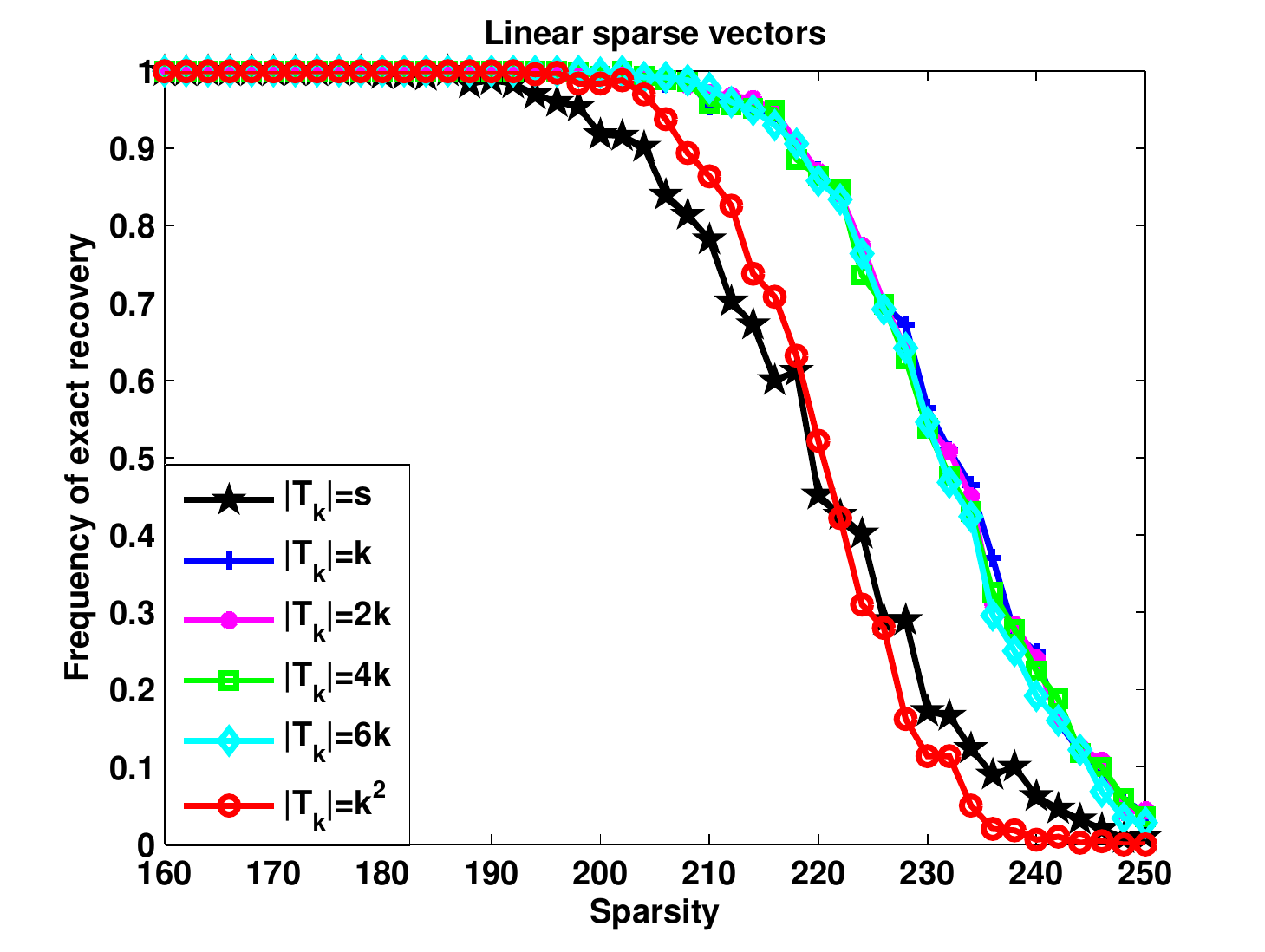}
&\includegraphics[width=6cm]{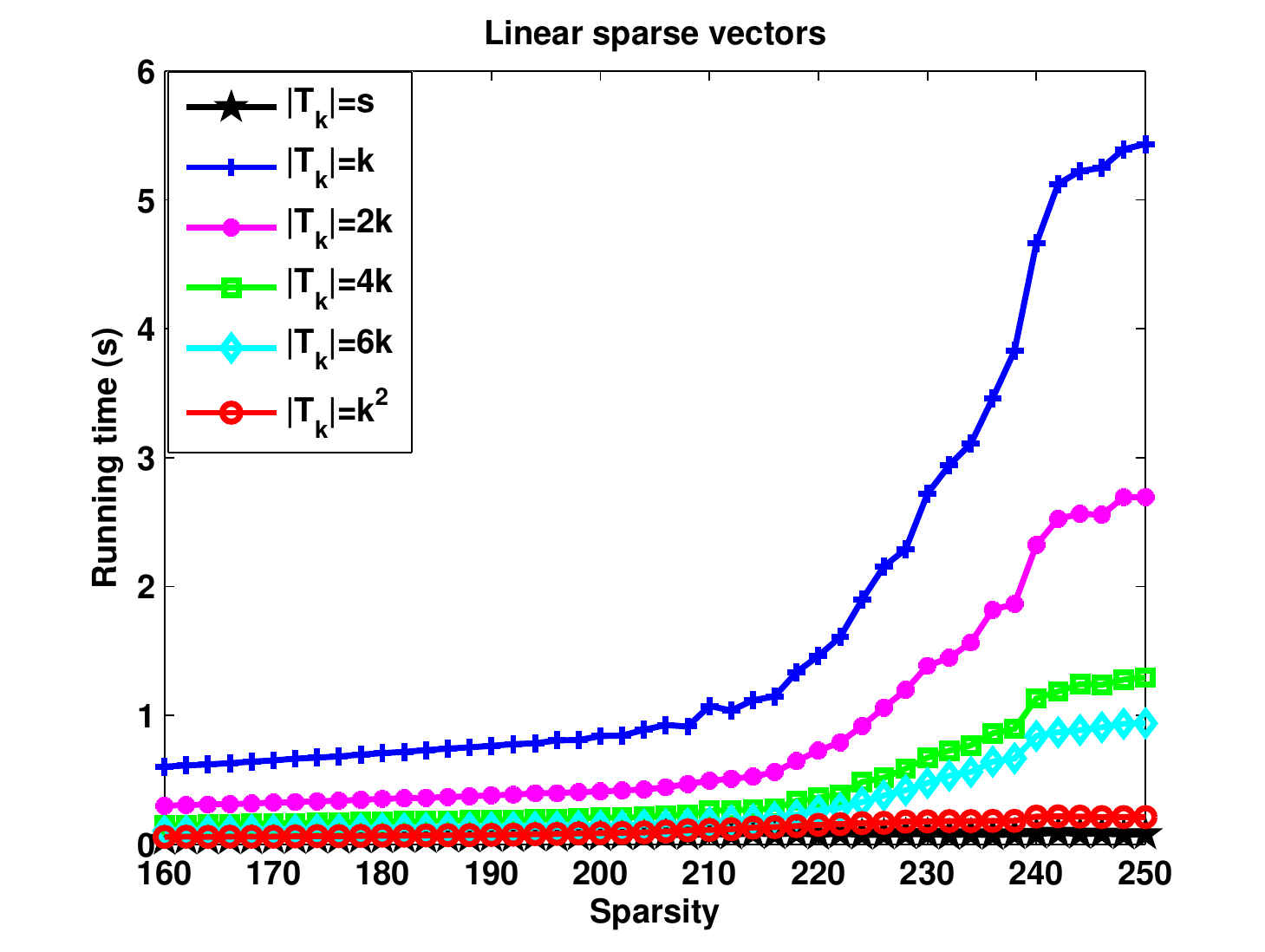}\\
\end{tabular}
\caption{Top (left to right): Frequency of successful recoveries using Gaussian sparse vectors. Running time using Gaussian sparse vectors. Middle (left to right): Frequency of successful recoveries of Bernoulli sparse vectors. Running time using Bernoulli sparse vectors. Bottom (left to right): Frequency of successful recoveries using linear sparse vectors. Running time using linear sparse vectors.}
\label{challengefig5}
\end{figure}
Note that AIHT and CGIHT need a prior estimation of the sparsity level, while GOMP, GHTP and AdptNST+HT+$f$-FB are more alike with increasing sizes of the index through iterations. GOMP is a generalization of OMP \cite{hongk} in the sense that multiple $P$ indices are identified at each iteration, where the value $P$ should not exceed $\sqrt{M}$. For a fair comparison, the particular index selection $|T_{k}|=k^{2}$ are used for GHTP and AdptNST+HT+$f$-FB in view of execution time and recovery accuracy. Since experiments focus on performance comparisons of greedy algorithms, we must comment that the comparisons are far from complete. In addition, the state of the art greedy algorithms, e.g. CoSaMP, are not included because that the running time of these algorithms are more than one order of magnitude higher than that of compared algorithms presented. The associated matlab codes can be downloaded from the authors' webpages or provided by authors in personal communication. A matlab implementation of the proposed algorithm is also available at \url{https://www.dropbox.com/s/uoh9sisbnwpy6ef/AdptNST%2BHT%2Bf-FB.zip?dl=0}.

Two performance metrics are used throughout the experiments. The first metric refers to the frequency/rate of exact recovery. An exact recovery is recorded whenever $\|x^{n}-x\|_{2}/\|x\|_{2}\leq10^{-4}$. Each algorithm is tested for $500$ (random) trials for every value of sparsity $s$. The second metric is the execution-time. The normalized mean square error (NMSE) is employed to evaluate robustness of algorithms and it is calculated by averaging normalized squared errors $\|x-\widehat{x}\|_{2}^{2}/\|x\|_{2}^{2}$ of $500$ independent trials, where $\widehat{x}$ denotes the estimate of the original signal $x$. The measurement matrix $A$ is an $500\times1000$ Gaussian random matrix with standard i.i.d. Gaussian entries.  The support of a sparse signal is also chosen randomly. The nonzero entries of Gaussian sparse signals are drawn independently from the Gaussian distribution with zero mean and unit variance, while the ones of the Bernoulli sparse signals are drawn independently from $\pm1$ with equiprobability and nonzero entries of $s$-sparse linear sparse vectors are $x_{j}=(s+1-j)/s$ for $j\in\{1,\ldots,s\}$.

\subsection{Performance comparison within the class of AdptNST+HT+$f$-FB algorithms}
We first study the effect of the number of indices $|T_k|=f(k)$ selected per iteration of the AdptNST+HT+$f$-FB algorithm.   For this experiment, the sparsity $s$ of Gaussian and linear sparse vectors varies from $160$ to $250$.  The matrix $A$ has again $M=500$ rows and $N=1000$ columns.

As shown in the top and last row of {\em Figure 1}, the frequencies of exact recovery of all algorithms are almost identical when $s\leq190$. However, when $s>190$, the performance of AdptNST+HT+$f$-FB with $|T_{k}|=s$, i.e., NST+HT+FB, degrades notably.  Interestingly, AdptNST+HT+$f$-FB algorithms without the prior estimation of $s$ offer enhanced performance. The second column of {\em Figure 1} also plots the execution-time.  We see that the efficiency of AdptNST+HT+$f$-FB with $|T_{k}|=k^{2}$ is tremendously improved. The running time of AdptNST+HT+$f$-FB with $|T_{k}|=k^{2}$
\begin{figure}[H]
\centering
\begin{tabular}{cc}
\includegraphics[width=6cm]{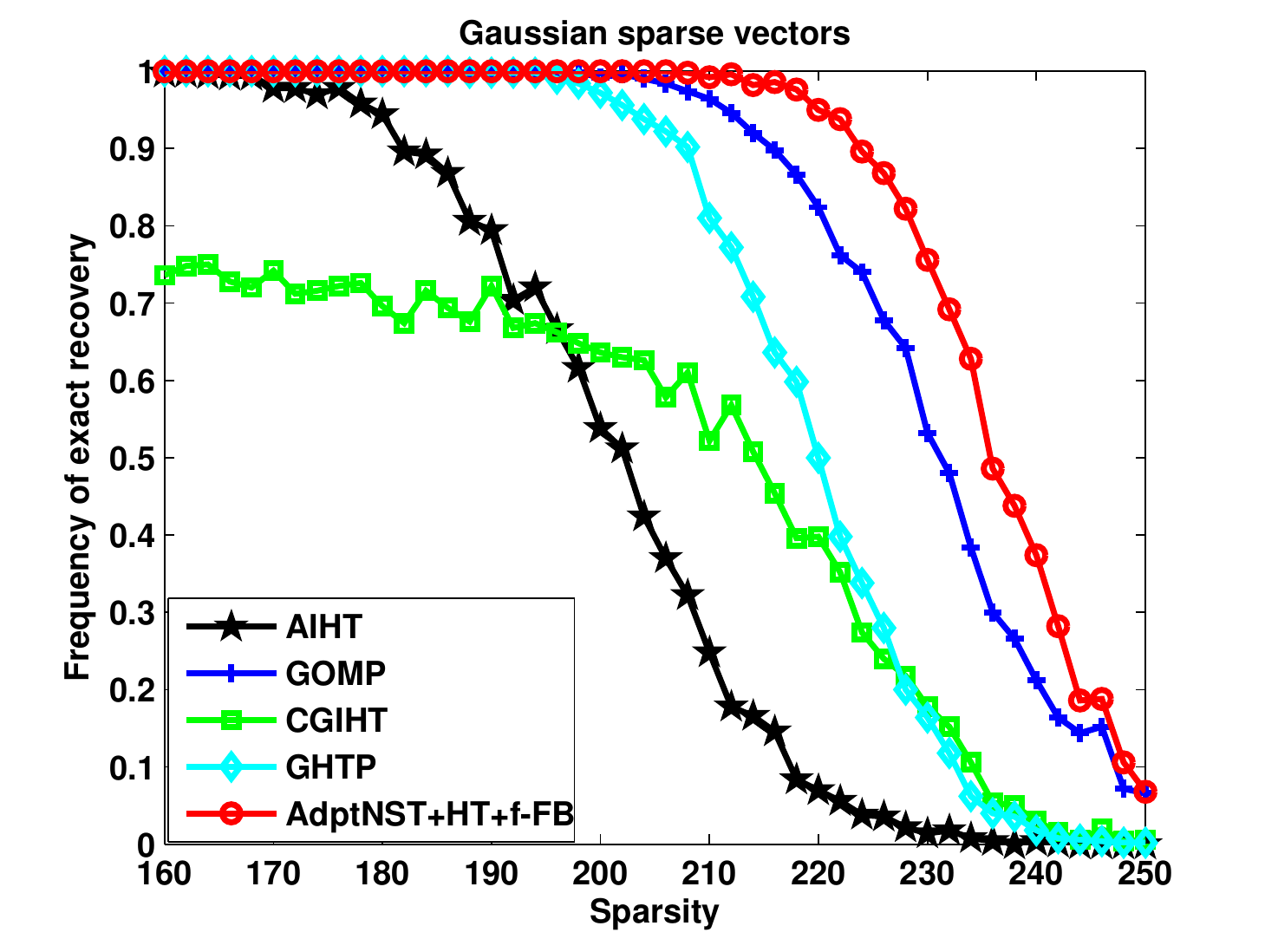}
&\includegraphics[width=6cm]{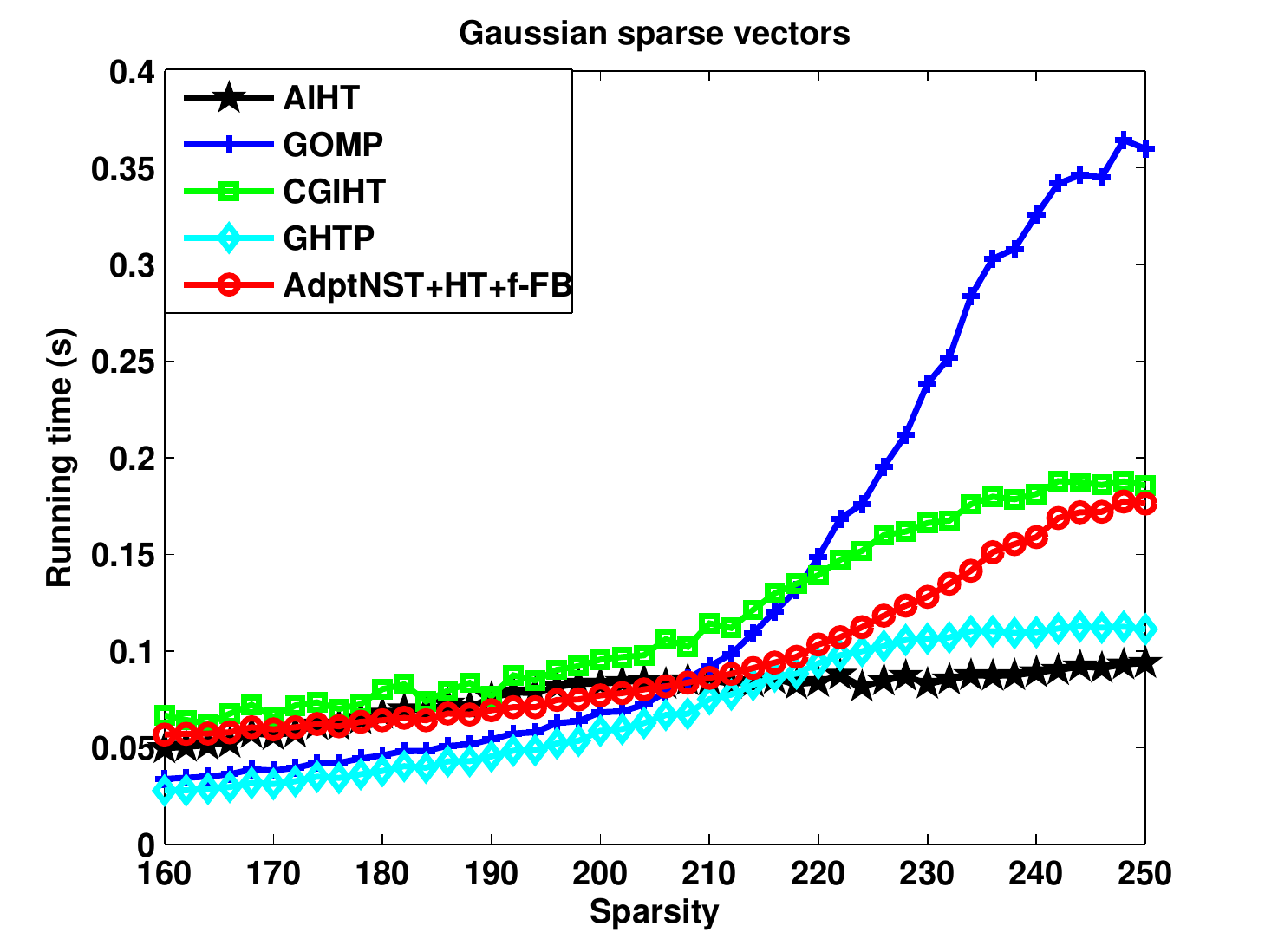}\\
\includegraphics[width=6cm]{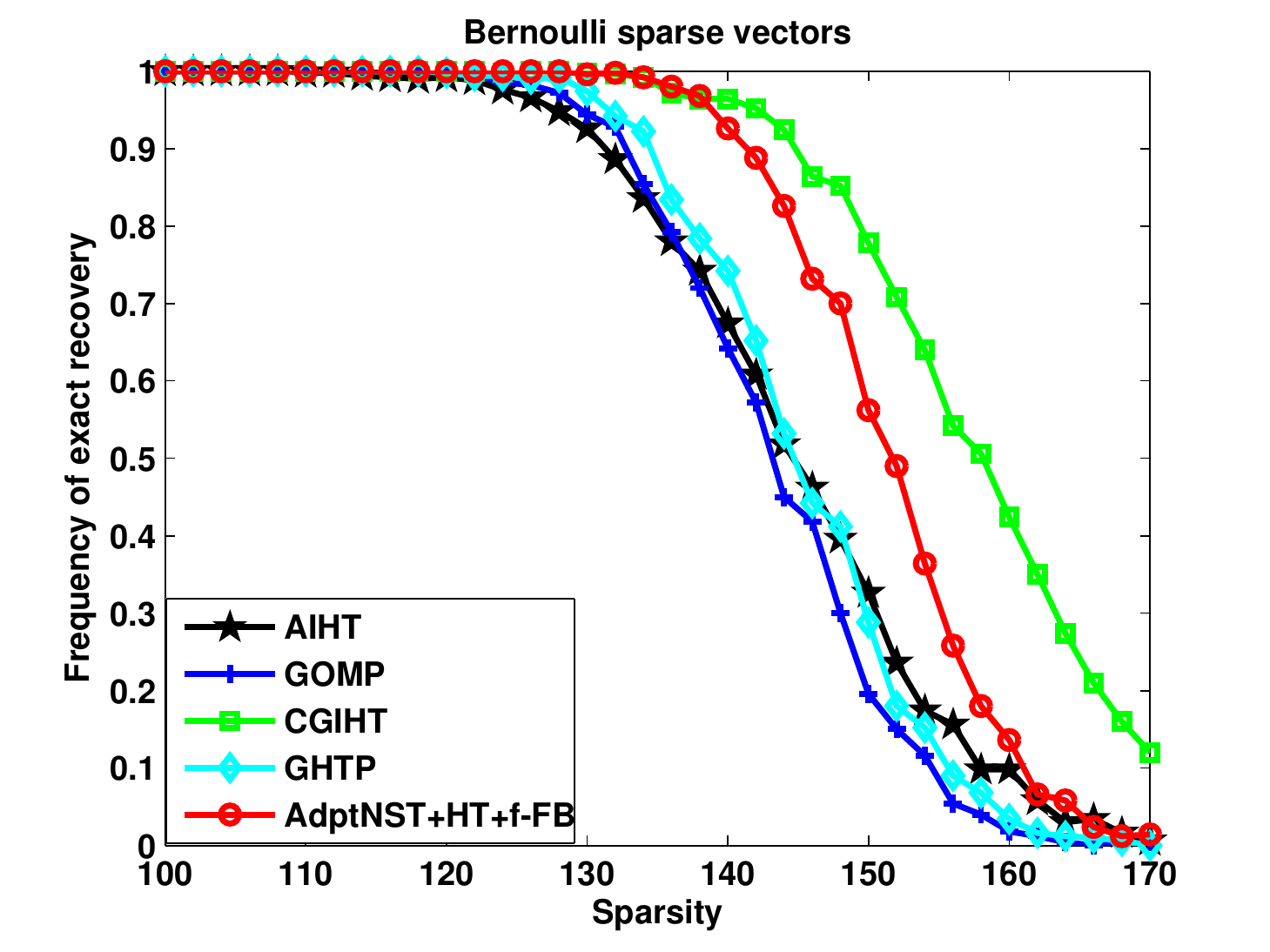}
&\includegraphics[width=6cm]{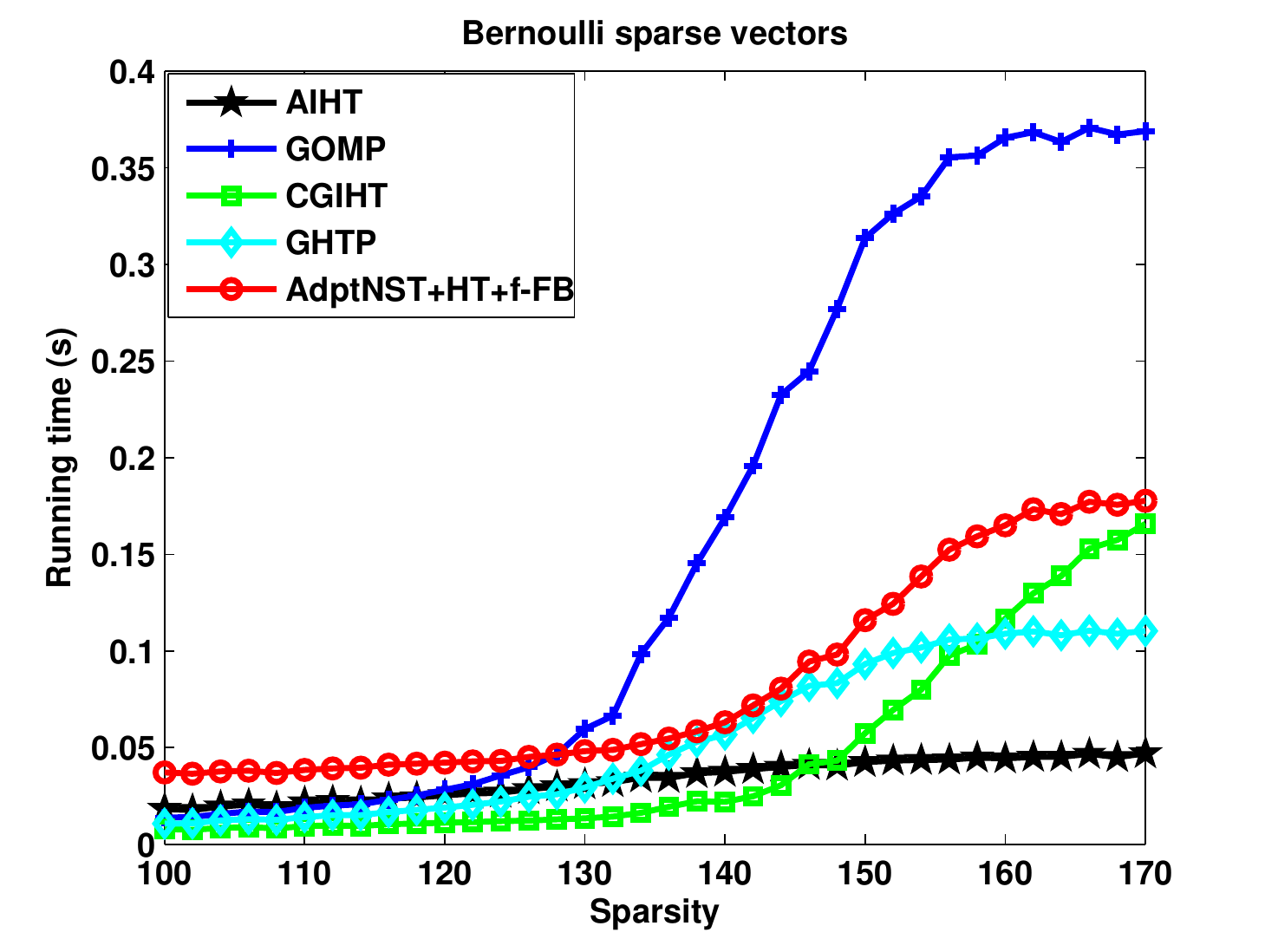}\\
\includegraphics[width=6cm]{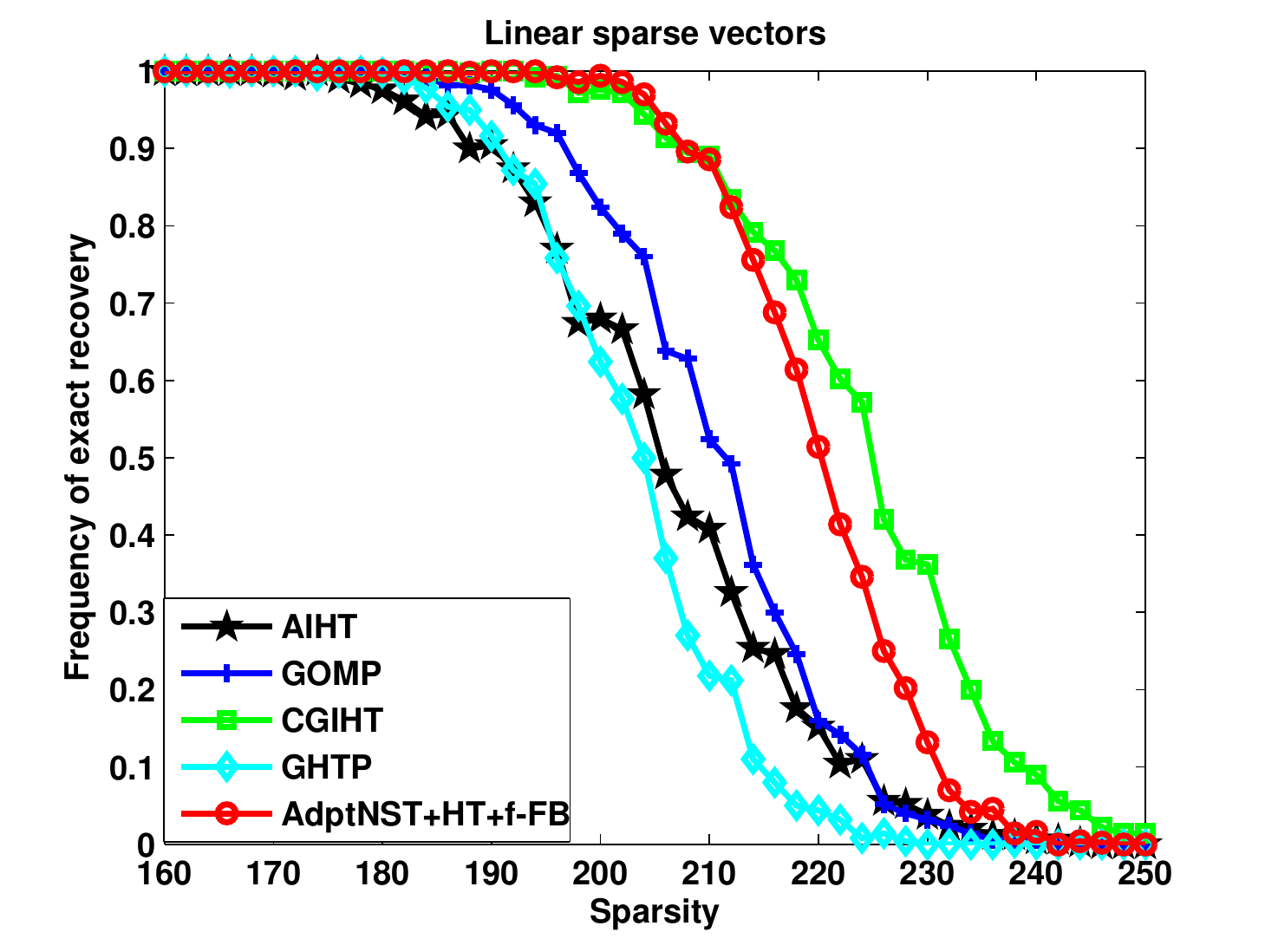}
&\includegraphics[width=6cm]{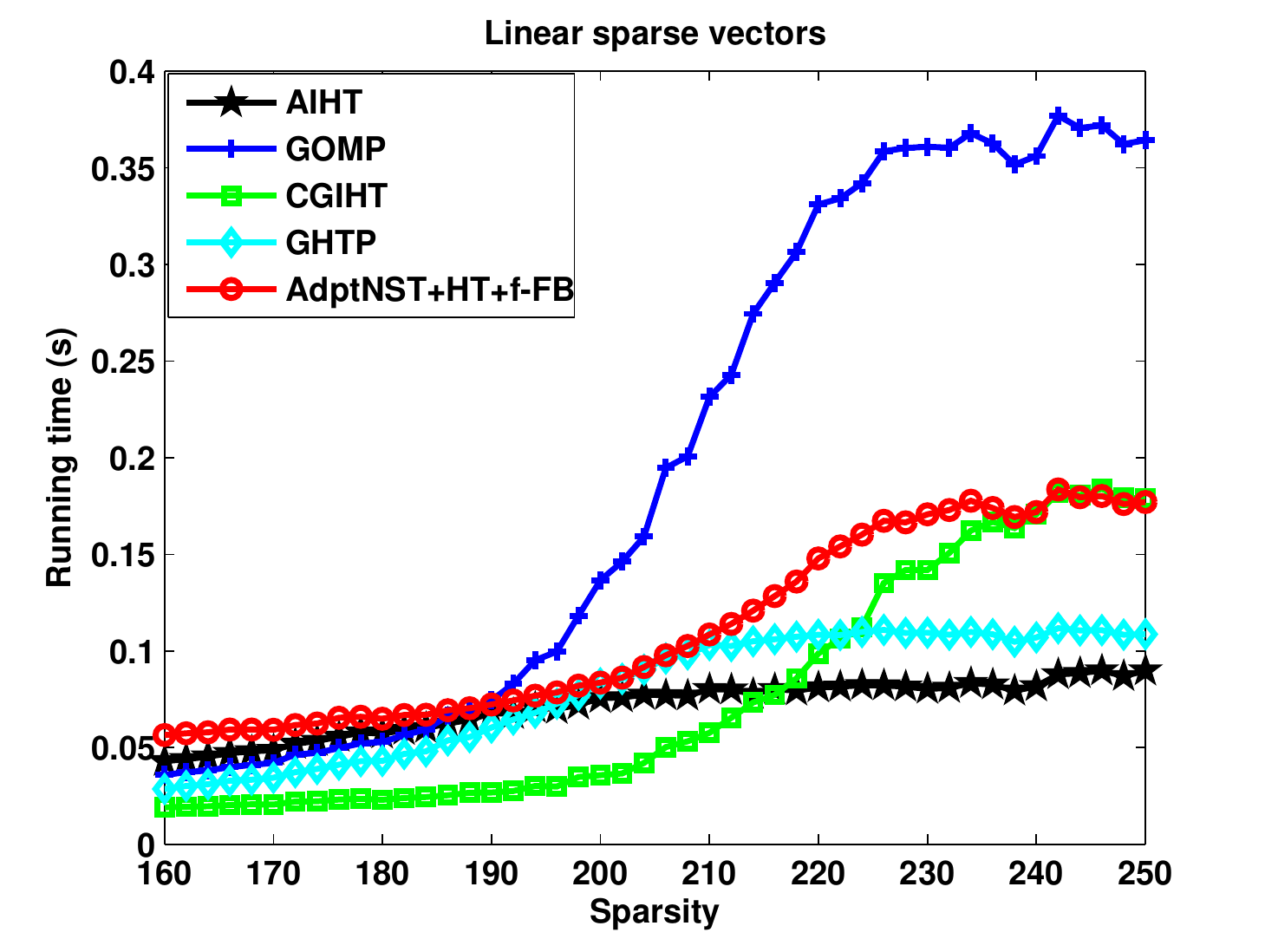}\\
\end{tabular}
\caption{Top (left to right): Frequency of successful recoveries using Gaussian sparse vectors. Running time using Gaussian sparse vectors. Middle (left to right): Frequency of successful recoveries using Bernoulli sparse vectors. Running time using Bernoulli sparse vectors. Bottom (left to right): Frequency of successful recoveries using linear sparse vectors. Running time using linear sparse vectors.}
\label{challengefig5}
\end{figure}
\noindent  is more than two orders-of-magnitude faster than that $|T_{k}|=6k$, and ten orders-of-magnitude faster than $|T_{k}|=k$. It verifies the claim that one can accelerate the convergence speed of the class of AdptNST+HT+$f$-FB algorithms by adjusting the size
of the support per iteration.

The middle row shows reconstructed results with Bernoulli sparse signals with sparsity varying from $120$ to $170$. It is evident that all cases present same performance when $s<140$ and NST+HT+FB outperforms all other cases when $s\geq 140$. Referring to running time, setting $|T_{k}|=k^{2}$ still outperforms other cases of AdptNST+HT+$f$-FB algorithms.  For other choices of linear $f(k)$, AdptNST+HT+$f$-FB with $|T_{k}|=k$, $|T_{k}|=2k$, $|T_{k}|=4k$, and $|T_{k}|=6k$ obtain better recovery than AdptNST+HT+$f$-FB with $|T_{k}|=k^{2}$ and NST+HT+FB.   However, AdptNST+HT+$f$-FB with $|T_{k}|=k^{2}$ and NST+HT+FB are typically faster than other cases.

Both Gaussian, Bernoulli and linear random signals demonstrate the benefit of selecting an appropriate number of indices per iteration. These comprehensive comparison points to the AdptNST+HT+$f$-FB scheme with $|T_{k}|=k^{2}$ as one significant approach with the balanced best performance.

\subsection{Overall comparison with known state-of-the-art greedy algorithms}

Presented here are comparisons among our AdptNST+HT+$f$-FB and state-of-the-art techniques such as AIHT, GOMP, CGIHT, and GHTP in terms of frequency of exact recovery and running time. Gaussian random sparse vectors are tested first.

In this experiment, sparsity varies from $160$ to $250$ as well with the same matrix $A$ of $M=500$ rows and $N=1000$ columns. It can be seen that AdptNST+HT+$f$-FB outperforms all other algorithms by a great margin in successful recovery frequencies. For the execution-time comparison, AIHT, CGIHT, GHTP and AdptNST+HT+$f$-FB are in the same level, which are all better than GOMP.

Experimental tests (Bernoulli and linear random sparse vectors) show that AdptNST+HT+$f$-FB still delivers reasonable performance better than that of AIHT, GOMP, CGIHT, and GHTP, though slightly underperforms that of CGIHT.

It is worth noting, however, that AdptNST+HT+$f$-FB does not require a prior knowledge on the sparsity $s$, whereas CGIHT still does.  These experiments suggest that no algorithm is consistently superior for all types of measurement matrices and sparse vectors.

But AdptNST+HT+$f$-FB is observed to have obviously advantageous balance of efficiency, adaptivity and accuracy compared with other algorithms.
\begin{figure}[H]
\centering
\begin{tabular}{ccc}
\includegraphics[width=5.2cm]{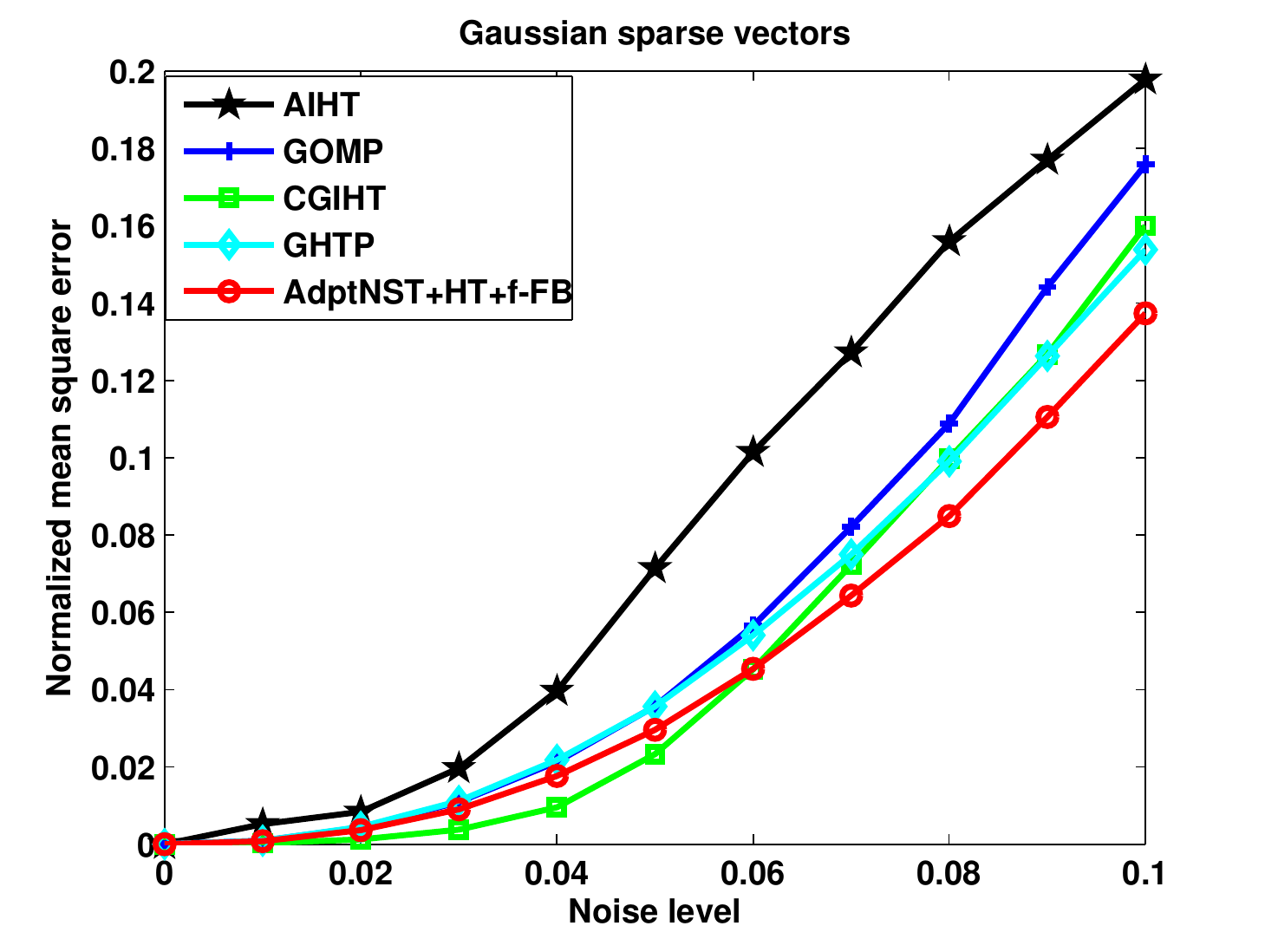}
&\includegraphics[width=5.2cm]{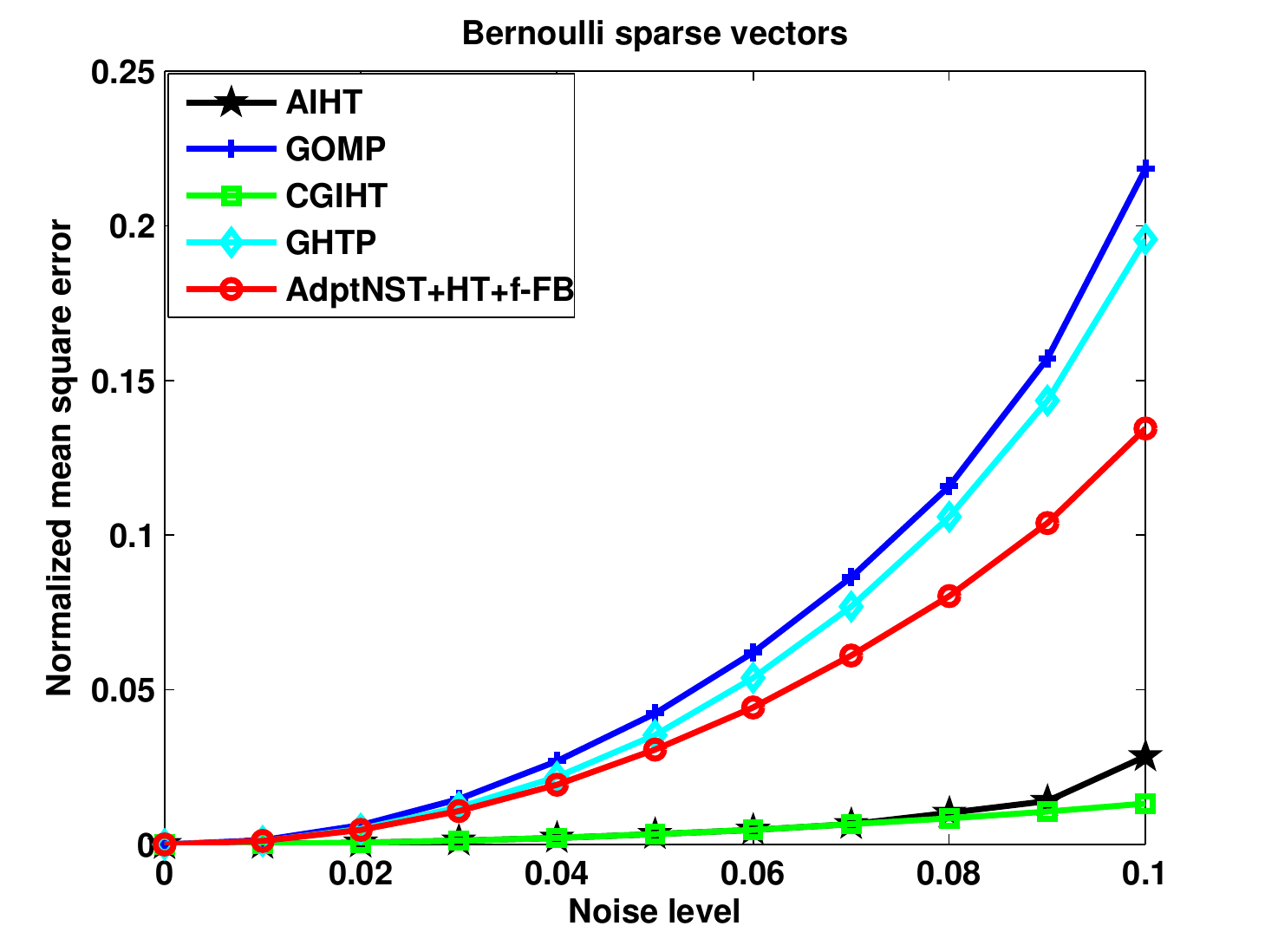}
&\includegraphics[width=5.2cm]{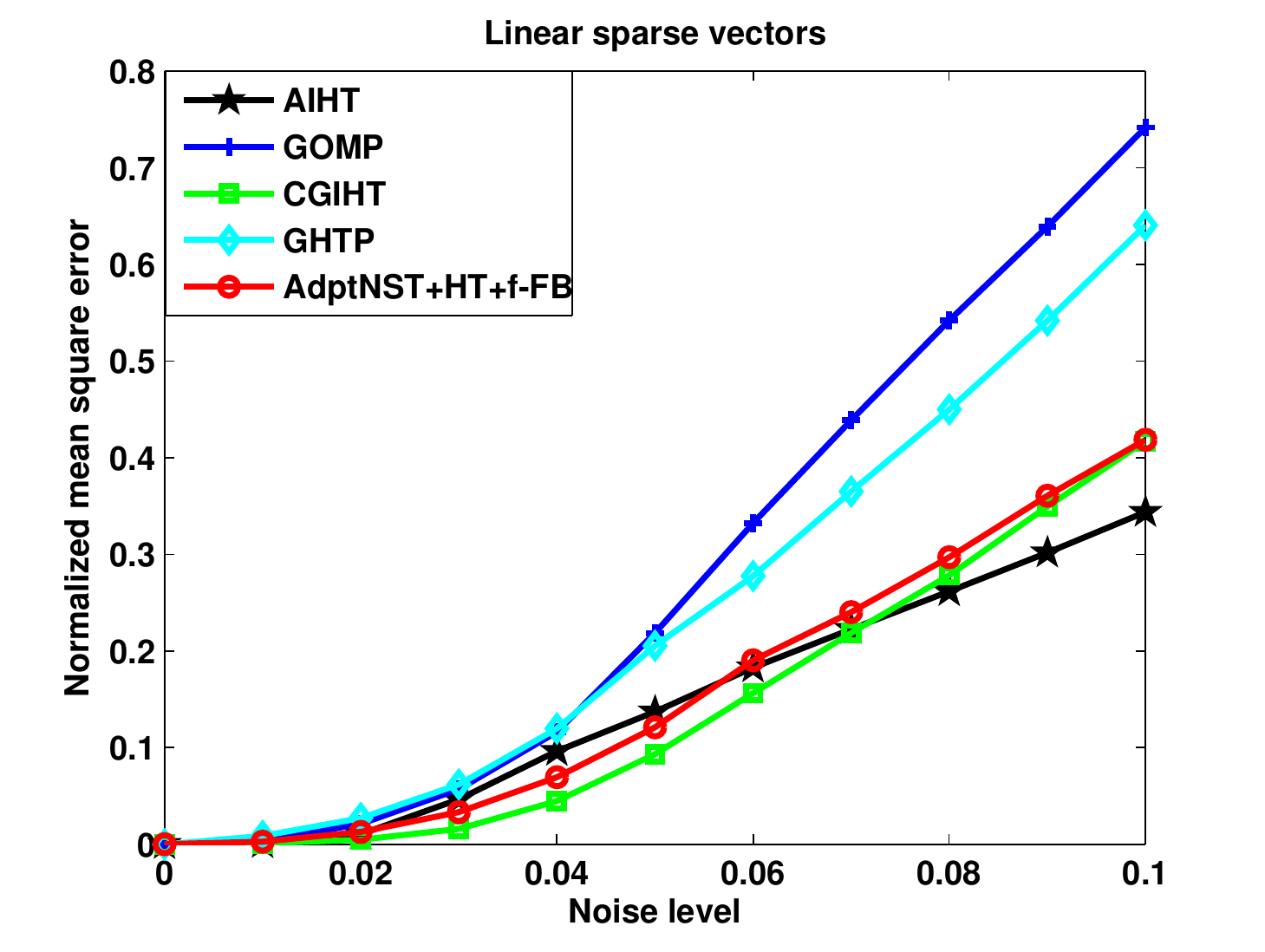}\\
\end{tabular}
\caption{(left to right): Normalized mean square error using Gaussian sparse vectors. Normalized mean square error using Bernoulli sparse vectors. Normalized mean square error using linear sparse vectors.}
\label{challengefig55}
\end{figure}
\subsection{Comparison in robustness to noise}
This section is to investigate performances of tested algorithms with noisy measurements.  There are three numerical experiments carried out. The first is to deal with the contaminated Gaussian sparse signals with fixing parameters $M=500$, $N=1000$, $s=200$ and varying standard deviation of zero-mean white Gaussian noise from $0$ to $0.1$. As shown in the {\em Figure \ref{challengefig55}}, AdptNST+HT+$f$-FB posses the best stabilities.

The second experiment is to recover Bernoulli sparse signals with fixing parameters $M=500$, $N=1000$, $s=120$. It is clear that CGIHT and AIHT outperforms other algorithms.

The third experiment is to examine the performance of recovering linear sparse vectors with the same parameters as the Gaussian case, it can be observed that AIHT, CGIHT and AdptNST+HT+$f$-FB deliver better performance than other algorithms.

These experiments show that the proposed class of AdptNST+HT+$f$-FB algorithms is superior to other state-of-the-art algorithms with balanced robustness to noise, and their adaptivity for requiring no prior knowledge about the sparsity $s$.

\subsection{Comparison applied to super-resolution: complex measurement matrix}
In \cite{HanNv}, the authors develop a mathematical theory of super-resolution, which aims to recover the high end of spectrum of an object from coarse scale information, i.e., samples from the low end of spectrum. Assume that a signal of interest can be represented as a superposition of Dirac measures
\begin{figure}[H]
\centering
\begin{tabular}{ccc}
\includegraphics[width=5.2cm]{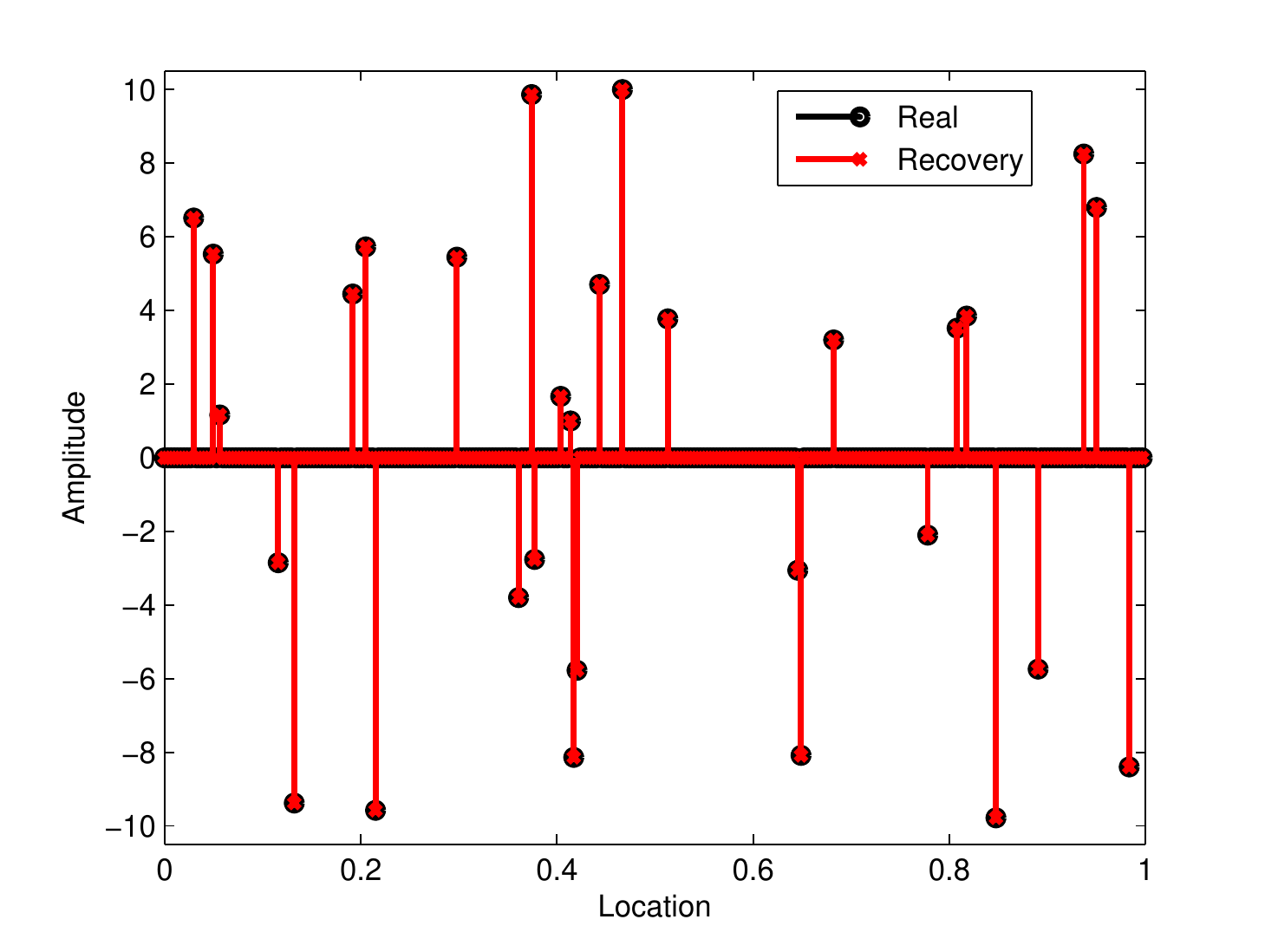}
&\includegraphics[width=5.2cm]{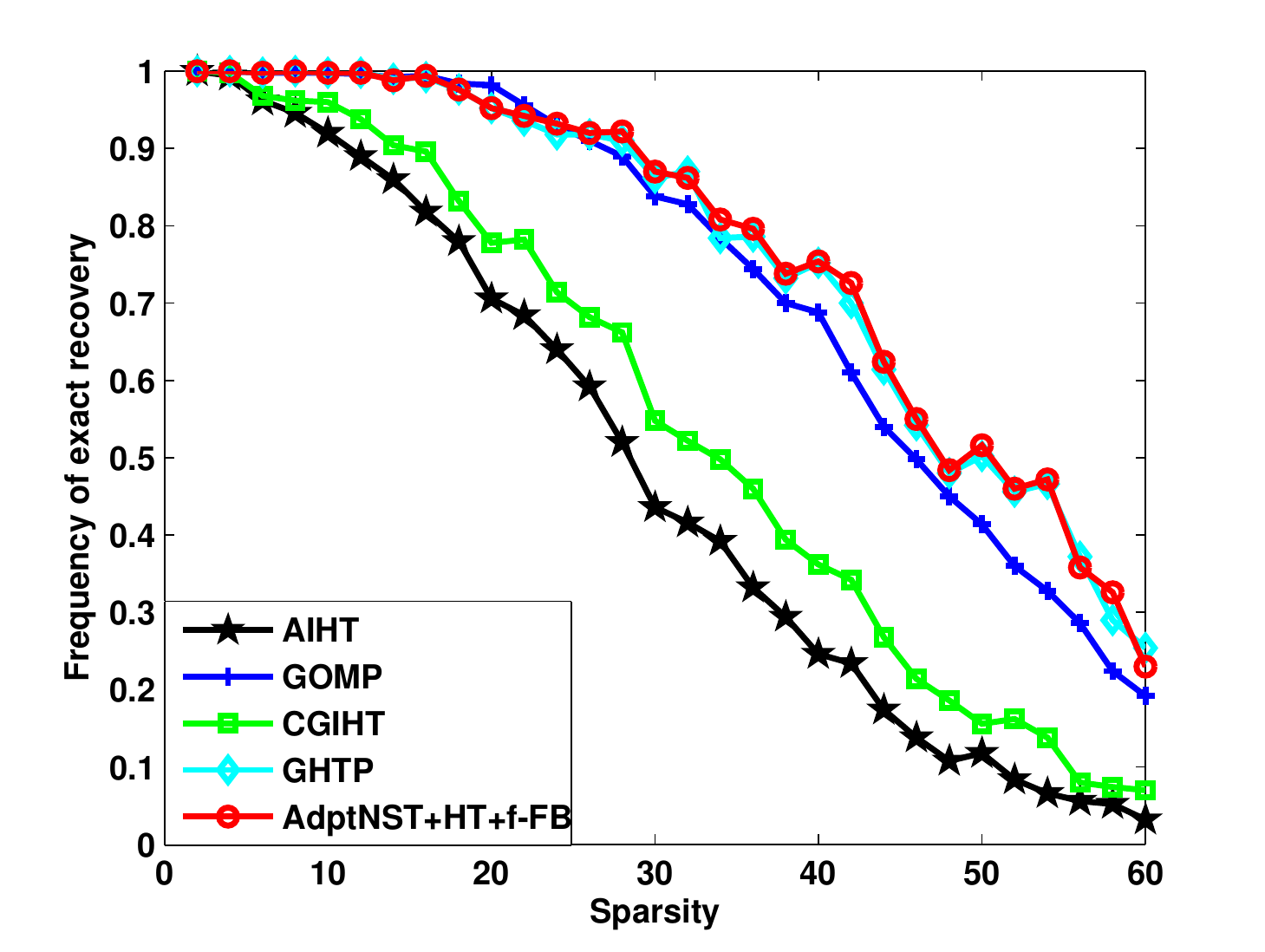}
&\includegraphics[width=5.2cm]{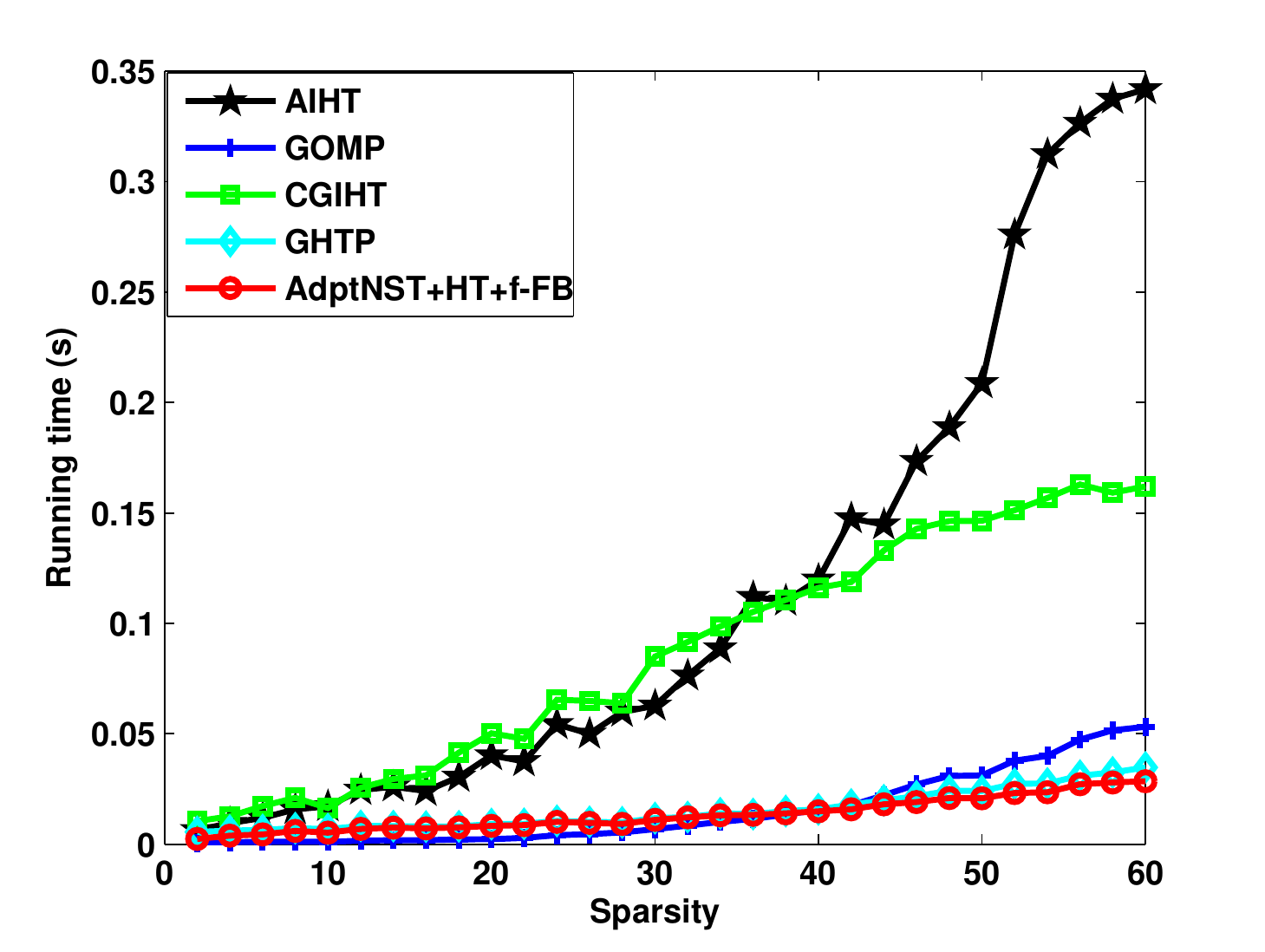}\\
\end{tabular}
\caption{(left to right): Schematic illustration of the super-resolution.  Frequency of successful recoveries using complex measurement matrix. Running time using complex measurement matrix.}
\label{challengefig5}
\end{figure}
\begin{equation}\label{supe}
x(t)=\sum\limits_{j=1}^{J}c_{j}\delta_{t_{j}},
\end{equation}
where $\delta_{t_{j}}$ is a Dirac measure at the location $t_{j}$, and $c_{j}$ is the corresponding amplitude. We want to recover the exact locations and amplitudes from samples of $x(t)$. The Fourier transform of (\ref{supe}) can be written as
\begin{equation}\label{wow}
y(\omega)=\sum\limits_{j=1}^{J}c_{j}\text{e}^{-i2\pi\omega t_{j}},
\end{equation}
which is in fact the problem of spectral estimation in signal processing.

Given the number of locations $J$, we approximate the locations $\{t_{j}\}$ by a subset that belongs to discrete grids $\mathcal {G}=\{g_{1},\ldots,g_{N}\}$ and denote amplitudes as $x\in\mathbb{C}^{N}$. It follows that the nonzero entries of $x$ equal to amplitudes $\{c_{j}\}$ when the grid point in $\mathcal {G}$ approaches the location $t_{j}$. Let $y=(y(\omega_{m}))\in\mathbb{C}^{M}$ be the data vector and the measurement matrix be $A\in\mathbb{C}^{M\times N}$ with $A_{mn}=(\text{e}^{-i2\pi\omega_{m} g_{n}})_{m=1}^{M}$, $n=1,\ldots,N$. The spectral estimation model would have the same linear equation
\begin{equation}\label{superre}
y=Ax,
\end{equation}
which in turn can be cast as sparse signal recover problems with complex measurement matrix $A$.

We first studied an example using AdptNST+HT+$f$-FB to illustrate the implementation of the algorithm in spectral estimation problems.  The frequencies $\{t_{j}\}$ $(j=1,\ldots,30)$ in (\ref{wow}), i.e., the locations in (\ref{supe}) are uniformly generated over $(0,1)$. The length of $y$ is $200$ and the number of grids $\mathcal {G}$ over $[0,1]$ is $300$, i.e., $A\in\mathbb{C}^{200\times 300}$. As shown in {\em Figure \ref{challengefig5}} (see the first figure), our algorithm can recover the exact locations and amplitudes. Furthermore, for comparison, we fix the parameters at $M=200$ and $N=300$. The experiments are carried out by varying the number of locations $J$ from $2$ to $60$. We can see from {\em Figure \ref{challengefig5}} that the performances of all algorithms degrade with increasing the number of locations and the complex measurement matrix seems to be unstable compared to the real case. Besides, AdptNST+HT+$f$-FB and GHTP deliver the best performances in terms of recovery accuracy and execution-time.

\section{Conclusions}
A class of adaptive iterative thresholding algorithms for signal reconstruction is thoroughly studied based on null space tuning, hard thresholding, and $f$-feedbacks. Analytical convergence analysis and proofs of uniform convergence of the algorithms are carried out.  Studies show that by selecting an appropriate number of indices per iteration, the convergence is significantly improved without reducing the recovery accuracy.  The theoretical findings are demonstrated and supported by extensive numerical simulations. Moreover, the experimental results show that the AdptNST+HT+$f$-FB algorithms have obviously advantageous balance of efficiency, adaptivity and accuracy compared with other state-of-the-art greedy iterative algorithms.

\section*{Acknowledgments}
 This work was partially supported by the National Natural Science Foundation of China (Grant Nos.61379014), and the NSF of USA (DMS-1313490, DMS-1615288).

\end{document}